\documentclass[a4paper,USenglish]{lipics-v2018}
\usepackage{comment}
\usepackage[utf8]{inputenc}
\usepackage[OT4]{fontenc}
\usepackage{graphicx}
\usepackage[svgnames]{xcolor}
\usepackage{tikz}
\usetikzlibrary{shapes}
\usepackage{amssymb}
\usepackage{amsmath}
\usepackage{amsfonts}
\usepackage{amsthm}
\usepackage{wrapfig}
\usepackage{algorithm}
\usepackage[export]{adjustbox}
\usepackage[noend]{algpseudocode}
\usepackage{algorithmicx}
\usepackage{thm-restate}
\usepackage{thmtools}
\usepackage{microtype}
\usepackage{multicol}
\usepackage{calc}

\newcommand{\removep}{\texttt{removeP}}
\newcommand{\removes}{\texttt{removeS}}
\newcommand{\remover}{\texttt{removeR}}

\makeatletter
\def\relaxtheorem#1{%
\expandafter\let\csname#1\endcsname\relax
\expandafter\let\csname c@#1\endcsname\relax
}
\makeatother

\relaxtheorem{theorem}
\relaxtheorem{lemma}
\relaxtheorem{observation}
\relaxtheorem{definition}
\relaxtheorem{corollary}
\relaxtheorem{example}
\relaxtheorem{remark}

\declaretheorem[name={Theorem},style=plain]{theorem}
\declaretheorem[name={Lemma},sibling=theorem]{lemma}
\declaretheorem[name={Observation},sibling=theorem]{observation}

\declaretheorem[name={Definition},sibling=theorem,style=definition]{definition}

\newcommand{\dual}[1]{\ensuremath{#1^{*}}}
\newcommand{\fv}[1]{\ensuremath{#1^{\diamond}}}

\newcommand{\contr}[2]{\ensuremath{#1 / #2}}
\newcommand{\remove}[2]{\ensuremath{#1 - #2}}

\newcommand{\abs}[1]{\left\lvert #1 \right\rvert}

\newcommand{\set}[1]{\left\{ #1 \right\}}
\newcommand{\sizeof}[1]{\left\lvert #1 \right\rvert}
\newcommand{\cond}{\mathrel{}\middle\vert\mathrel{}}
\newcommand{\eid}{\ensuremath{\mathrm{id}}}

\author{Jacob Holm}{University of Copenhagen, Denmark}{jaho@di.ku.dk}{http://orcid.org/0000-0001-6997-9251}{Jacob Holm is supported by Mikkel Thorup's Advanced Grant DFF-0602-02499B from the Danish Council for Independent Research under the Sapere Aude research career programme.}
\author{Giuseppe F. Italiano}{University of Rome Tor Vergata, Italy}{giuseppe.italiano@uniroma2.it}{https://orcid.org/0000-0002-9492-9894}{Giuseppe F. Italiano is partially supported by the Italian Ministry of Education, University and Research under Project AMANDA (Algorithmics for MAssive and Networked DAta).}
\author{Adam Karczmarz}{University of Warsaw, Poland}{a.karczmarz@mimuw.edu.pl}{https://orcid.org/0000-0002-2693-8713}{Adam Karczmarz is supported by the grants 2014/13/B/ST6/01811 and 2017/24/T/ST6/00036 of the Polish National Science Center.}
\author{Jakub Łącki}{Google Research, USA}{jlacki@google.com}{https://orcid.org/0000-0001-9347-0041}{When working on this paper Jakub Łącki was partly supported by the EU FET project MULTIPLEX no. 317532 and the Google Focused Award on "Algorithms for Large-scale Data Analysis" and Polish National Science Center grant number 2014/13/B/ST6/01811. Part of this work was done while Jakub Łącki was visiting the Simons Institute for the Theory of Computing.}
\author{Eva Rotenberg}{Technical University of Denmark, Denmark}{erot@dtu.dk}{http://orcid.org/0000-0001-5853-7909}{}

\titlerunning{Decremental SPQR-trees for Planar Graphs}
\authorrunning{J. Holm, G.\,F. Italiano, A. Karczmarz, J. Łącki, and E. Rotenberg}
\Copyright{Jacob Holm, Giuseppe F. Italiano, Adam Karczmarz, Jakub Łącki, and Eva Rotenberg}
\subjclass{, 
, 
}

\hideLIPIcs
\nolinenumbers

\begin{document}

\title{Decremental SPQR-trees for Planar Graphs}

\maketitle

\begin{abstract}
We present a decremental data structure for maintaining the SPQR-tree of a planar graph subject to edge contractions and deletions. The update time, amortized over $\Omega(n)$ operations, is $O(\log^2 n)$.

Via SPQR-trees, we give a decremental data structure for maintaining $3$-vertex connectivity in planar graphs. It answers queries in $O(1)$ time and processes edge deletions and contractions in $O(\log^2 n)$ amortized time.
This is an exponential improvement over the previous best bound of $O(\sqrt{n}\,)$ that has stood for over 20 years. In addition, the previous data structures only supported edge deletions.
\end{abstract}

\section{Introduction}

A graph algorithm is called \emph{dynamic} if it is able to answer queries about a given property while the graph is undergoing a sequence of updates, such as edge insertions and deletions. It is \emph{incremental} if it handles only insertions,
\emph{decremental} if it handles only deletions,
and \emph{fully dynamic} if it handles both insertions and deletions.
In designing dynamic graph algorithms, one is typically interested in achieving fast query times (either constant or polylogarithmic),
while minimizing the update times. The
ultimate goal is to perform fast both queries and updates, i.e., to have both query and update times either constant or polylogarithmic.
So far, the quest for obtaining polylogarithmic time algorithms has been successful only in few cases.
Indeed, efficient dynamic algorithms with \emph{polylogarithmic} time per update are known
only for few problems, such as dynamic connectivity,
$2$-connectivity,  minimum spanning tree and maximal matchings in undirected graphs
(see, e.g., \cite{bgs2015,HeTh97,HoLiTh01,Holm15,KaKiMo13,S16,Thorup00, Nilsen13}).
On the other hand, some dynamic problems appear to be inherently  harder. For example, the fastest known algorithms for basic dynamic problems, such as reachability,
transitive closure, and dynamic shortest paths have only \emph{polynomial} times per update
(see, e.g., \cite{CHILP16,DI04,DI08,K99,RZ08,S04,T05}).

A similar situation holds for planar graphs where dynamic problems have been studied extensively,
see e.g.~\cite{Abraham2012, Diks2007,  Eppstein96, Eppstein92,  Giammarresi:96, Gustedt98, 2017arXiv170610228H, INSW11, scc-decomposition, steiner-tree, Lacki2011, decremental-connectivity,  Sub-ESA-93}.
Despite this long-time effort, the best algorithms known for some basic problems on planar graphs,
such as dynamic shortest paths and dynamic planarity testing, 
still have polynomial update time bounds.
For instance, for fully dynamic shortest paths on planar graphs the best known bound per operation
is\footnote{Throughout the paper, we use the notation 
$\widetilde{O}(f(n))$ to hide polylogarithmic factors.} $\widetilde{O}(n^{2/3})$
amortized~\cite{FR06,INSW11,KMNS12,K05}, while for fully dynamic planarity testing the best known bound per operation is $O(\sqrt{n}\,)$ amortized~\cite{Eppstein96}.

In the last years, this exponential gap between polynomial and polylogarithmic bounds has sparkled some new exciting research. On one hand, it was shown that there are dynamic graph problems,
including fully dynamic shortest paths, fully dynamic single-source reachability and fully dynamic strong connectivity, 
 for which it may be difficult to achieve subpolynomial update bounds. 
 This started with the pioneering work by Abboud and Vassilevska-Williams~\cite{AW14}, who proved 
conditional lower bounds
based on popular conjectures. Very recently, Abboud and Dahlgaard~\cite{AD16} proved polynomial update time lower bounds
for dynamic shortest paths also on planar graphs, again based on popular conjectures.

On the other hand, the question of improving the 
update bounds from polynomial to polylogarithmic, has, for several other dynamic graph problems, received much attention in the last years. For instance, there was a very recent improvement from polynomial to polylogarithmic bounds for decremental single-source reachability (and strongly connected components) on planar graphs: more precisely, the improvement was from $O(\sqrt{n}\,)$ amortized~\cite{scc-decomposition} to $O(\log^2 n\log\log n)$ amortized~\cite{IKLS17} (both amortizations are over sequences of $\Omega(n)$ updates). Other problems that received a lot of attention are fully dynamic connectivity and minimum spanning tree in general graphs.
Up to very recently, the best worst-case bound for both problems was $O(\sqrt{n}\,)$ per update~\cite{EGIN97}: since then, much effort has been devoted towards improving this bound (see e.g., ~\cite{KaKiMo13,KR16,NS17,NSW17,W17}). 

In this paper, we follow the ambitious goal of achieving polylogarithmic update bounds for dynamic graph problems.
In particular, we show how to improve the update times from polynomial to polylogarithmic for another important problem on planar graphs: decremental $3$-vertex connectivity. Given a graph $G=(V,E)$ and two vertices $x,y\in V$ we say that $x$ and $y$ are $2$-vertex connected (or, as we say in the following, \emph{biconnected}) if there are at least two vertex-disjoint paths between $x$ and $y$ in $G$. We say that $x$ and $y$ are
$3$-vertex connected (or, as we say in the following, \emph{triconnected}) if there are at least three vertex-disjoint paths between $x$ and $y$ in $G$. The decremental planar triconnectivity problem consists of maintaining a planar graph $G$ subject to an arbitrary sequence of edge deletions, edge contractions, and query operations which test whether two arbitrary input vertices are triconnected. 
We remark that decremental triconnectivity on planar graphs is of particular importance. Apart from being a fundamental graph property, a triconnected planar graph has only one planar embedding, a property which is heavily used in graph drawing, planarity testing and testing for isomorphism~\cite{Hopcroft73, Hopcroft74, kant2001}. %
Furthermore, our extended repertoire of operations, which includes edge contractions, contains all operations needed to obtain a graph minor, which is another important notion for planar graphs. 

While polylogarithmic update bounds for decremental $2$-edge and $3$-edge  connectivity, and for decremental biconnectivity on planar graphs have been known for more than two decades~\cite{Giammarresi:96}, decremental triconnectivity on planar graphs presents some special challenges. 
Indeed, while
connectivity cuts for $2$-edge and $3$-edge connectivity, and for biconnectivity have simple counterparts in the dual graph or in the vertex-face graph (see Section~\ref{sec:preliminaries} for a formal definition of vertex-face graph), triconnectivity cuts (separation pairs, i.e., pairs of vertices whose removal disconnects the graph) have a much more complicated structure in planar graphs. Roughly speaking, maintaining $2$-edge and $3$-edge connectivity cuts in a planar graph under edge deletions corresponds to maintaining respectively self-loops and cycles of length 2 (pairs of parallel edges) in the dual graph under edge contractions. On the other side, maintaining biconnectivity and triconnectivity cuts in a planar graph under edge deletions corresponds to maintaining, respectively, cycles of length 2 and cycles of length 4 in the vertex-face graph.
While detecting cycles of length 2 boils down to finding duplicates in the multiset of all edges, detecting cycles of length 4 under edge contractions is far more complex.
We believe that this is the reason why designing a fast solution for decremental triconnectivity on planar graphs has been an elusive goal, and the best bound known of $O(\sqrt{n}\,)$ per update~\cite{Eppstein:1998} has been standing for over two decades.

\subparagraph*{Our results and techniques.} In this paper, we show how to solve the decremental triconnectivity problem on planar graphs in constant time per query and  $O(\log ^2 n)$ amortized time per edge deletion or contraction, over any sequence of $\Omega (n)$ deletions and contractions. This is an exponential speed-up over the previous $O(\sqrt{n}\,)$ long-standing bound~\cite{Eppstein:1998}. To obtain our bounds,  we also need to solve decremental biconnectivity on planar graphs in constant time per query and  $O(\log ^2 n)$ amortized time per edge deletion or contraction.
(A better $O(\log n)$ amortized bound can be obtained if no contractions are allowed~\cite{2017arXiv170610228H}).
Our results are achieved with the help of two new tools, which may be of independent interest. 

The first tool is an algorithm for efficiently detecting and reporting cycles of length 4 as they arise in a dynamic planar embedded graph subject to edge contractions and insertions. The algorithm works for a graph with bounded face-degree, i.e, where each face is delimited by at most some constant number of edges. Specifically, given a plane embedded graph with bounded face-degree subject to edge-contractions and edge-insertions across a face, after each dynamic operation we can report all edges that lie on a length-$4$ cycle because of this dynamic operation. 
The total running time is $O(n \log n)$. One of the challenges that we face is that a planar graph may have as many as $\Omega(n^2)$ distinct cycles of length $4$.
Still, we give a surprisingly simple algorithm for solving this problem.
The difficulty of the algorithm lies in the analysis --- in fact, this analysis is the most technically involved part of this paper.
 
The second tool is a new data structure that maintains the SPQR-tree~\cite{Battista:96} of a planar graph, while the graph is updated with edge deletions and edge contractions, %
 in $O(\log ^2 n)$ amortized time per operation. While incremental algorithms for maintaining the SPQR tree were known for more than two decades~\cite{Battista:96,BT96}, 
 to the best of our knowledge 
no decremental algorithm was previously known.

\subparagraph{Organization of the paper.}
The remainder of the paper is organized as follows.
In Section~\ref{sec:preliminaries}, we introduce notation and definitions that we later use.
Then, in Section~\ref{sec:overview} we present a high-level overview of our results.
Section~\ref{sec:SPQR} presents our new algorithm for maintaining an SPQR-tree during edge deletions and contractions.

Due to space constraints, an algorithm for detecting cycles of length $4$ under contractions, which is a key tool in maintaining an SPQR tree is described in Appendix~\ref{sec:4cycle}.
Moreover, 
the detailed discussion of how to use the SPQR-trees to maintain information about triconnectivity
is deferred to Appendix~\ref{sec:3vertex}.
Finally, 
the proofs omitted from Section~\ref{sec:SPQR} are given in Appendix~\ref{sec:SPQR-proofs}.

\section{Preliminaries}\label{sec:preliminaries}
Throughout the paper we use the term \emph{graph} to denote an undirected \emph{multigraph}, that is we allow the graphs to have parallel edges and self-loops.
Formally, each edge $e$ of such a graph is a pair $(\{u,w\},\eid(e))$ consisting of a pair of vertices and a unique \emph{integer identifier} used to distinguish between the parallel edges.
For simplicity, in the following we skip the identifier and use just $uw$ to denote one of the edges connecting vertices $u$ and $w$.
If the graph contains no parallel edges and no self-loops, we call it \emph{simple}.

Given a graph $G$, we use $V(G)$ to denote the vertices, and $E(G)$ to denote the edges of $G$.  For any $X\subseteq V(G)$ let $G[X]$ denote the subgraph $(X,\set{(\set{u,v},l)\in E(G)\cond u,v\in X})$ of $G$ \emph{induced by} $X$.

The \emph{components} of a graph $G$ are the minimal subgraphs $H\subseteq G$ such that for every edge $uv\in E(G)$, $u\in V(H)$ if and only if $v\in V(H)$.  The components of a graph partition the vertices and edges of the graph.  A graph $G$ is \emph{connected} if it consists of a single component.  For a positive integer $k$, a graph is \emph{$k$-vertex connected} if and only if it is connected, has at least $k$ vertices, and stays connected after removing any set of at most $k-1$ vertices.  The \emph{local vertex connectivity} of a pair of vertices $u$, $v$, denoted $\kappa(u,v)$, is the maximal number of internally vertex-disjoint $u,v$-paths.  By Menger's Theorem~\cite{Menger:1927}, $G$ is $k$-vertex connected if and only if $\kappa(u,v)\geq k$ for every pair of non-adjacent vertices $u,v$. We say that $u$, $v$ are (locally) $k$-vertex connected if $\kappa(u,v)\geq k$.
We follow the common practice of using \emph{biconnected} as a synonym for $2$-vertex connected and \emph{triconnected} as a synonym for $3$-vertex connected.
An articulation point $v$ of $G$ is a vertex whose removal disconnects $G$. Thus a graph is biconnected if and only if it has no articulation points.%

Let $G$ be a graph and $e \in E(G)$.
We use $\remove{G}{e}$ to denote the graph obtained from $G$ by removing~$e$.
If $e$ is not a self-loop, we use $\contr{G}{e}$ to denote the graph obtained by contracting~$e$.
A \emph{cycle} $C$ of length $\sizeof{C}=k$ in a graph $G$ is a cyclic sequence of edges $C = e_1, e_2 \ldots, e_k$ where $e_i=u_iu_{i+1}$ for $1\leq i<k$ and $e_k=u_ku_1$.  A cycle is \emph{simple} if $\eid(e_i)\neq \eid(e_j)$ and $u_i\neq u_j$ for $i\neq j$. We sometimes abuse notation and treat a  cycle as a set of edges or a cyclic sequence of vertices.
Note that this definition allows cycles of length $1$ (a self-loop) or $2$ (a pair of parallel edges).
Let $G$ be a planar embedded graph. For each component $H$ of $G$, let
 $\dual{H}$ denote the dual graph of $H$, defined as the graph obtained by creating a vertex for each face in the embedding of $H$, and an edge $\dual{e}$ (called the \emph{dual edge} of $e$), connecting the two (not necessarily distinct) faces that $e$ is incident to.  Let $\dual{G}$ denote the graph obtained from $G$ by taking the dual of each component.

Each face $f$ in a planar embedded graph is bounded by a (not necessarily simple) cycle called the \emph{face cycle} for $f$. We call the length of this cycle the \emph{face-degree} of $f$. We call any other cycle a \emph{separating cycle}.

Let $G$ be a connected planar embedded multigraph with at least one edge.
Define the set $\fv{E}(G)$ of \emph{corners}\footnote{For alternative definitions, see e.g.~\cite{Holm2017} and~\cite{NYAS:NYAS340}. The latter uses the name \emph{angles} for what we call corners.} of $G$ to be the the set of ordered pairs of (not necessarily distinct) edges $(e_1,e_2)$ such that $e_1$ immediately precedes $e_2$ in the clockwise order around some vertex, denoted $v(e_1,e_2)$.  Note that if $(e_1,e_2)\in \fv{E}(G)$, then $(\dual{e_2},\dual{e_1})\in \fv{E}(\dual{G})$.
We denote by $\fv{G}$ %
the \emph{vertex-face} graph\footnote{A.k.a. the \emph{vertex-face incidence graph}~\cite{doi:10.1137/0406017}, the \emph{angle graph}~\cite{NYAS:NYAS340}, and the \emph{radial graph}~\cite{ARCHDEACON199237}%
.} of $G$ (see Figure~\ref{fig:facevertexgraph}). This is a plane embedded multigraph with vertex set $V(G)\cup V(\dual{G})$, and an edge between $v(e_1,e_2)$ and $v(\dual{e_2},\dual{e_1})$ for each corner $(e_1,e_2)\in \fv{E}(G)$.
Abusing notation slightly, we can write $\fv{G}$ as $=(V(G)\cup V(\dual{G}),\fv{E}(G))$.
We use the following well-known facts about the vertex-face graph:
\begin{enumerate}
\item $\fv{G}$ is bipartite and planar, with a natural embedding given by the embedding of $G$.
\item\label{fact:fvSameDual} The vertex-face graphs of $G$ and $\dual{G}$ are the same: $\fv{G} = \fv{(\dual{G})}$.
\item\label{fact:fvEdges} There is a one-to-one correspondence between the edges of $G$ and the faces of $\fv{G}$ (in the natural embedding, each face of $\fv{G}$ contains exactly one edge of $G$ interior, see Fig~\ref{fig:facevertexgraph}).
\item\label{fact:fvDual4regular} $\dual{(\fv{G})}$ (also known as the \emph{medial graph}) is $4$-regular.
\item\label{fact:fv2conn} $\fv{G}$ is simple if and only if $G$ is loopless and biconnected (See e.g.~\cite[Theorem 5(i)]{Brinkmann:2005:GSQ:2651845.2652314}).
\item\label{fact:fv3conn} $\fv{G}$ is simple, triconnected and has no separating $4$-cycles if and only if $G$ is simple and triconnected (See e.g.~\cite[Theorem 5(iv)]{Brinkmann:2005:GSQ:2651845.2652314}).
\end{enumerate}

\begin{figure}
\centering
\begin{adjustbox}{max width=0.4\linewidth}
\begin{tikzpicture}[y=0.80pt, x=0.80pt, yscale=-1.000000, xscale=1.000000, inner sep=0pt, outer sep=0pt]
\clip (0,0) rectangle (200.16078mm, 120.87504mm);
\begin{scope}[
        shift={(-65.29457,-375.84708)},
        every path/.style={
          draw=black,
          line join=miter,
          line cap=butt,
          even odd rule,
          line width=0.800pt,
          },
        every node/.style={
          draw,
          circle,
          fill = black,
          inner sep=0pt,
          minimum size=3mm,
          }
        ]
  \coordinate (v1) at (179.5971,633.0936);
  \coordinate (v2) at (261.0256,503.0936);
  \coordinate (v3) at (361.0256,608.8079);
  \coordinate (v4) at (296.7399,735.9507);
  \coordinate (v5) at (295.3114,625.9507);
  \coordinate (v6) at (259.5971,597.3793);
  \coordinate (v7) at (450.1685,500.2364);
  \coordinate (v8) at (548.7399,503.0936);
  \coordinate (v9) at (457.3114,597.3793);
  \coordinate (v10) at (525.8828,671.6650);
  \coordinate (v11) at (435.8828,667.3793);
  \coordinate (v12) at (657.3113,580.2364);
  \coordinate (v13) at (645.8828,461.6650);
  \coordinate (v14) at (608.7399,507.3793);

  \path (v1) -- (v2) -- (v3) -- (v4) -- (v5) -- (v6) -- cycle;
  \path (v1) -- (v4);
  \path (v6) -- (v2);
  \path (v5) -- (v3);
  \path (v3) -- (v7) -- (v8) -- (v9) -- (v10) -- (v11) -- cycle;
  \path (v7) -- (v9) -- (v11);
  \path (v8) -- (v12) -- (v13) -- (v14);
  \path (v13) -- (v8) -- (v14) -- (v12);
  \path (v10) -- (v12);
  \node at (v1) {};
  \node at (v2) {};
  \node at (v3) {};
  \node at (v4) {};
  \node at (v5) {};
  \node at (v6) {};
  \node at (v7) {};
  \node at (v8) {};
  \node at (v9) {};
  \node at (v10) {};
  \node at (v11) {};
  \node at (v12) {};
  \node at (v13) {};
  \node at (v14) {};
\end{scope}
\end{tikzpicture}
 \end{adjustbox}%
\hspace{0.15\linewidth}%
\begin{adjustbox}{max width=0.4\linewidth}
\begin{tikzpicture}[y=0.80pt, x=0.80pt, yscale=-1.000000, xscale=1.000000, inner sep=0pt, outer sep=0pt]
\clip (0,0) rectangle (200.16078mm, 120.87504mm);
\begin{scope}[
        shift={(-65.29457,-375.84708)},
        ]
  \coordinate (v1) at (179.5971,633.0936);
  \coordinate (v2) at (261.0256,503.0936);
  \coordinate (v3) at (361.0256,608.8079);
  \coordinate (v4) at (296.7399,735.9507);
  \coordinate (v5) at (295.3114,625.9507);
  \coordinate (v6) at (259.5971,597.3793);
  \coordinate (v7) at (450.1685,500.2364);
  \coordinate (v8) at (548.7399,503.0936);
  \coordinate (v9) at (457.3114,597.3793);
  \coordinate (v10) at (525.8828,671.6650);
  \coordinate (v11) at (435.8828,667.3793);
  \coordinate (v12) at (657.3113,580.2364);
  \coordinate (v13) at (645.8828,461.6650);
  \coordinate (v14) at (608.7399,507.3793);

  \coordinate (f1) at (298.1334,586.8804);
  \coordinate (f2) at (249.7070,655.8966);
  \coordinate (f3) at (318.9138,643.9219);
  \coordinate (f4) at (235.9927,583.6108);
  \coordinate (f5) at (199.2813,414.0233);
  \coordinate (f6) at (416.3991,597.7974);
  \coordinate (f7) at (469.4094,643.5790);
  \coordinate (f8) at (554.9488,583.3400);
  \coordinate (f9) at (599.5256,519.4867);
  \coordinate (f10) at (482.6620,527.9201);
  \coordinate (f11) at (630.8499,509.8484);
  \coordinate (f12) at (604.3448,490.5719);
\end{scope}
\begin{scope}[
        shift={(-65.29457,-375.84708)},
        every path/.style={
          draw=red,
          line join=miter,
          line cap=butt,
          miter limit=4.00,
          even odd rule,
          line width=2.400pt
          },
        every node/.style={
          draw,
          thin,
          circle,
          color=black,
          fill=white,
          inner sep=0pt,
          minimum size=3mm,
          }
        ]
  \path (f1) -- (v2);
  \path (f1) -- (v6);
  \path (f2) -- (v1);
  \path (f3) -- (v4);
  \path (v5) -- (f1) -- (v3) -- (f3) -- cycle;
  \path (v4) -- (f2) -- (v6);
  \path (f2) -- (v5);
  \path (f4) -- (v1);
  \path (v2) -- (f4) -- (v6);
  \path (f5) .. controls (237.7332,415.9853) and (361.0256,608.8079) .. (v3);
  \path (v2) -- (f5) -- (v1);
  \path (v3) .. controls (369.0889,1037.5459) and (-110.5481,527.3386) .. (f5);
  \path (v4) .. controls (174.4974,789.4023) and (52.1935,516.5597) .. (f5);
  \path (v7) .. controls (438.0851,490.5719) and (241.3204,413.0565) .. (f5);
  \path (v8) .. controls (528.4436,467.6811) and (212.7913,411.0565) .. (f5);
  \path (v13) .. controls (609.1639,425.5138) and (197.1292,411.0565) .. (f5);
  \path (v7) -- (f6) -- (v11) -- (f7) -- (v10) --
    (f8) -- (v8) -- (f9) -- (v12);
  \path (v3) -- (f6) -- (v9);
  \path (v7) -- (f10) -- (v9);
  \path (f10) -- (v8);
  \path (f7) -- (v9);
  \path (v12) -- (f8) -- (v9);
  \path (v12) -- (f11) -- (v13) -- (f12) -- (v8);
  \path (f11) -- (v14) -- (f12);
  \path (f9) -- (v14);
  \path (v12) .. controls (734.4610,445.9951) and (655.3551,348.8175) .. (f5);
  \path (v10) .. controls (1162.9597,391.4227) and (386.6479,328.9725) .. (f5);
  \path (v11) .. controls (170.5591,1086.3801) and (-100.5689,391.3523) .. (f5);

  \node at (f1) {};
  \node at (f2) {};
  \node at (f3) {};
  \node at (f4) {};
  \node at (f5) {};
  \node at (f6) {};
  \node at (f7) {};
  \node at (f8) {};
  \node at (f9) {};
  \node at (f10) {};
  \node at (f11) {};
  \node at (f12) {};
\end{scope}
\begin{scope}[
        shift={(-65.29457,-375.84708)},
        every path/.style={
          draw=black!25,
          line join=miter,
          line cap=butt,
          miter limit=4.00,
          even odd rule,
          line width=0.800pt
          },
        every node/.style={
          draw=black,
          circle,
          fill=black,
          inner sep=0pt,
          minimum size=3mm,
          }
        ]
  \path (v1) -- (v2) -- (v3) -- (v4) -- (v5) -- (v6) -- cycle;
  \path (v1) -- (v4);
  \path (v6) -- (v2);
  \path (v5) -- (v3);
  \path (v3) -- (v7) -- (v8) -- (v9) -- (v10) -- (v11) -- cycle;
  \path (v7) -- (v9) -- (v11);
  \path (v8) -- (v12) -- (v13) -- (v14);
  \path (v13) -- (v8) -- (v14) -- (v12);
  \path (v10) -- (v12);
  \node at (v1) {};
  \node at (v2) {};
  \node at (v3) {};
  \node at (v4) {};
  \node at (v5) {};
  \node at (v6) {};
  \node at (v7) {};
  \node at (v8) {};
  \node at (v9) {};
  \node at (v10) {};
  \node at (v11) {};
  \node at (v12) {};
  \node at (v13) {};
  \node at (v14) {};
\end{scope}
\end{tikzpicture}
 \end{adjustbox}
\caption{Left: a plane embedded graph. Right: the corresponding vertex-face graph (red) and the underlying graph (gray). }
\label{fig:facevertexgraph}
\end{figure}
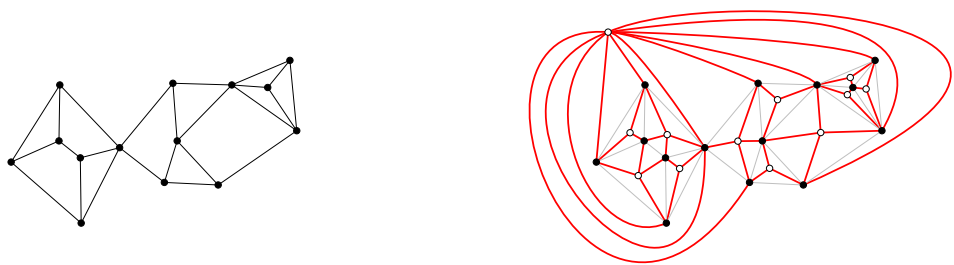

If $v$ is an articulation point in $G$ or has a self-loop, then in any planar embedding of $G$ there is at least one face $f$ whose face cycle contains $v$ at least twice. %
Any such $f$ is either an articulation point or has a self-loop in $\dual{G}$, and $v$ and $f$ are connected by (at least) two edges in $\fv{G}$.%

The dynamic operations on $G$ correspond to dynamic operations on $\dual{G}$ and $\fv{G}$.
Deleting a non-bridge edge $e$ of $G$ corresponds to contracting $\dual{e}$ in $\dual{G}$, that is $\dual{(\remove{G}{e})} = \contr{\dual{G}}{\dual{e}}$. Similarly, contracting an edge $e$ corresponds to deleting the corresponding edge from the dual, so $\dual{(\contr{G}{e})}=\remove{\dual{G}}{\dual{e}}$. Finally, deleting a non-bridge edge or contracting an edge corresponds to adding and then immediately contracting an edge across a face of $\fv{G}$ (and removing two duplicate edges).

The useful concept of a separation is well-defined, even for general graphs:
\begin{definition}\label{def:separation}
	Given a graph $G=(V,E)$, a \emph{separation} of $G$ is a pair of vertex sets $(V',V'')$ such that the induced subgraphs $G'=G[V'],G''=G[V'']$ cover $G$, and $V'\setminus V''$ and $V''\setminus V'$ are both nonempty. A separation is \emph{balanced} if $\max\set{\abs{V'},\abs{V''}\strut\!}\leq\alpha\abs{V}$ for some fixed constant $\frac{1}{2}\leq\alpha<1$. 
	If $(V',V'')$ is a separation of $G$, the set $S=V'\cap V''$ is called a \emph{separator} of $G$. A separator $S$ is \emph{small} if $\abs{S}=O(\sqrt{\abs{V}})$, and it is a \emph{cycle separator} if the subgraph of $G$ induced by $S$ is Hamiltonian.
\end{definition}

\section{Overview of Our Approach}\label{sec:overview}

Our data structure for decremental triconnectivity in planar graphs consists of two main ingredients. Before describing them, we need a few definitions.
We recall that a graph $G$ that is biconnected but not triconnected has at least one separation pair, i.e., a pair of vertices that can be removed to disconnect $G$:
\begin{definition}[Hopcroft and Tarjan~{\cite[p. 6]{hopcroft1973dividing}}]\label{def:separationpair}
  Let $\set{a,b}$ be a pair of vertices in a biconnected multigraph $G$. Suppose the edges of $G$ are divided into equivalence classes $E_1,E_2,\ldots,E_k$, such that two edges which lie on a common path not containing any vertex of $\set{a,b}$ except as an end-point are in the same class. The classes $E_i$ are called the \emph{separation classes} of $G$ with respect to $\set{a,b}$. If there are at least two separation classes, then $\set{a,b}$ is a \emph{separation pair} of $G$ unless (i) there are exactly two separation classes, and one class consists of a single edge, or (ii) there are exactly three classes, each consisting of a single edge (these two exceptions actually make it easier to state some properties related to separation pairs).

\end{definition}

Note that \emph{separation pair}, which is a pair of vertices, should not be confused with \emph{separation} (see Definition~\ref{def:separation}), which is a pair of vertex sets.

Our first ingredient for decremental triconnectivity is an algorithm for efficiently detecting separation pairs in planar graphs.
The second ingredient is the maintenance of the SPQR-tree~\cite{Battista:96} for each biconnected component of a graph $G$ under edge deletions and contractions.
The SPQR-tree captures the structure of all separating pairs, and can be defined as follows:

\begin{figure}
  \begin{adjustbox}{max width=\textwidth}
\begin{tikzpicture}[y=0.80pt, x=0.80pt, yscale=-1.000000, xscale=1.000000, inner sep=0pt, outer sep=0pt,
  edge/.style={
    draw=black,
    line join=miter,
    line cap=butt,
    even odd rule,
    line width=0.800pt,
  },
  vertex/.style={
    fill=black,
    line join=miter,
    line cap=butt,
    line width=0.800pt,
    draw,
    circle,
    minimum size=3mm,
  },
]
  \clip (0,0) rectangle (297mm, 210mm);
  \begin{scope}[
      shift={(0,-308.26772)},
    ]
    \coordinate (p9) at (474,506);
    \coordinate (p10) at (642,635);
    \coordinate (p11) at (555,860);
    \coordinate (p12) at (341,815);
    \coordinate (p13) at (270,620);
    \coordinate (p14) at (262,395);
    \coordinate (p15) at (322,499);
    \coordinate (p16) at (140,783);
    \coordinate (p17) at (288,962);
    \coordinate (p22) at (830,679);
    \coordinate (p23) at (761,906);
    \coordinate (p24) at (687,826);
    \coordinate (p26) at (711,730);

    \path[edge] (p9) -- (p10) -- (p11) -- (p12) -- (p13) -- cycle;
    \path[edge] (p14) -- (p9) -- (p15) -- (p13) -- cycle;
    \path[edge] (p15) -- (p14);
    \path[edge] (p13) -- (p16) -- (p17) -- (p11);
    \path[edge] (p10) -- (p22) -- (p23) -- (p11) -- (p24) -- (p23);
    \path[edge] (p24) -- (p26) -- (p22);
    \path[edge] (p10) -- (p26);

    \node[vertex] at (p9) {};
    \node[vertex] at (p10) {};
    \node[vertex] at (p11) {};
    \node[vertex] at (p12) {};
    \node[vertex] at (p13) {};
    \node[vertex] at (p14) {};
    \node[vertex] at (p15) {};
    \node[vertex] at (p16) {};
    \node[vertex] at (p17) {};
    \node[vertex] at (p22) {};
    \node[vertex] at (p23) {};
    \node[vertex] at (p24) {};
    \node[vertex] at (p26) {};
  \end{scope}

\end{tikzpicture}

\begin{tikzpicture}[y=0.80pt, x=0.80pt, yscale=-1.000000, xscale=1.000000, inner sep=0pt, outer sep=0pt,
  graph edge/.style={
    draw=black,
    line join=miter,
    line cap=butt,
    even odd rule,
    line width=0.800pt,
  },
  virtual edge/.style={
    draw=black,
    dash pattern=on 4.80pt off 4.80pt,
    line join=miter,
    line cap=butt,
    miter limit=4.00,
    even odd rule,
    line width=0.800pt,
  },
  tree edge/.style={
    draw=red,
    line join=miter,
    line cap=butt,
    miter limit=4.00,
    even odd rule,
    line width=2.400pt,
  },
  font={\LARGE\bf},
]
  \clip (0,0) rectangle (297mm, 210mm);
  \begin{scope}[
      shift={(0,-308.26772)},
      tree node/.style={
        draw=red,
        line cap=round,
        miter limit=4.00,
        line width=1.545pt,
        circle,
        minimum size=12mm,
      },
      vertex/.style={
        draw=black,
        fill=black,
        minimum size=3mm,
        circle,
      }
    ]
    \coordinate (p1) at (372,824);
    \coordinate (p9) at (183,313);
    \coordinate (p10) at (396,425);
    \coordinate (p11) at (243,418);
    \coordinate (p12) at (189,538);
    \coordinate (p13) at (104,740);
    \coordinate (p14) at (12,851);
    \coordinate (p15) at (160,1030);
    \coordinate (p16) at (300,990);
    \coordinate (p21) at (844,730);
    \coordinate (p22) at (1031,774);
    \coordinate (p23) at (962,1001);
    \coordinate (p24) at (754,958);
    \coordinate (p25) at (888,921);
    \coordinate (p27) at (912,826);
    \coordinate (p33) at (720,690);
    \coordinate (p34) at (632,890);
    \coordinate (p35) at (724,831);
    \coordinate (p36) at (747,710);
    \coordinate (p37) at (735,682);
    \coordinate (p38) at (607,780);
    \coordinate (p39) at (557,856);
    \coordinate (p40) at (645,632);
    \coordinate (p41) at (494,518);
    \coordinate (p42) at (282,637);
    \coordinate (p43) at (240,585);
    \coordinate (p44) at (343,567);
    \coordinate (p45) at (429,485);
    \coordinate (p46) at (435,458);
    \coordinate (p47) at (287,479);
    \coordinate (p48) at (419,450);
    \coordinate (p57) at (304,744);
    \coordinate (p58) at (388,818);
    \coordinate (p59) at (389,925);
    \coordinate (p63) at (179,723);
    \coordinate (p64) at (275,772);
    \coordinate (p65) at (325,862);
    \coordinate (p66) at (335,956);
    \coordinate (p67) at (242,913);
    \coordinate (p68) at (191,823);
    \coordinate (p72) at (786,777);
    \coordinate (p73) at (758,855);
    \coordinate (p74) at (309,604);
    \coordinate (p75) at (419,560);
    \coordinate (p77) at (618,669);
    \coordinate (p78) at (668,699);
    \coordinate (p81) at (444,512);
    \coordinate (p82) at (505,596);
    \coordinate (p85) at (326,448);
    \coordinate (p86) at (338,472);
    \coordinate (p89) at (734,776);
    \coordinate (p90) at (804,777);
    \coordinate (p91) at (841,783);
    \coordinate (p92) at (868,807);
    \coordinate (p95) at (270,681);
    \coordinate (p96) at (374,639);
    \coordinate (p99) at (198,843);
    \coordinate (p100) at (161,845);
    \coordinate (p103) at (288,802);
    \coordinate (p104) at (318,790);
    \coordinate (p106) at (153,913);
    \coordinate (p107) at (254,833);
    \coordinate (p108) at (506,669);
    \coordinate (p109) at (681,792);
    \coordinate (p110) at (877,871);
    \coordinate (p111) at (332,521);
    \coordinate (p112) at (264,459);
    \coordinate (p113) at (265,720);
    \coordinate (p114) at (372,842);

    \node[vertex] at (p9) {};
    \node[vertex] at (p10) {};
    \node[vertex] at (p11) {};
    \node[vertex] at (p12) {};
    \node[vertex] at (p13) {};
    \node[vertex] at (p14) {};
    \node[vertex] at (p15) {};
    \node[vertex] at (p16) {};
    \node[vertex] at (p21) {};
    \node[vertex] at (p22) {};
    \node[vertex] at (p23) {};
    \node[vertex] at (p24) {};
    \node[vertex] at (p25) {};
    \node[vertex] at (p27) {};
    \node[vertex] at (p34) {};
    \node[vertex] at (p37) {};
    \node[vertex] at (p39) {};
    \node[vertex] at (p40) {};
    \node[vertex] at (p41) {};
    \node[vertex] at (p42) {};
    \node[vertex] at (p43) {};
    \node[vertex] at (p46) {};
    \node[vertex] at (p57) {};
    \node[vertex] at (p58) {};
    \node[vertex] at (p59) {};
    \node[vertex] at (p63) {};
    \node[vertex] at (p66) {};

    \path[graph edge] (p9) -- (p10) -- (p11) -- (p12) -- cycle;
    \path[graph edge] (p11) -- (p9);
    \path[graph edge] (p13) -- (p14) -- (p15) -- (p16);
    \path[graph edge] (p21) -- (p22) -- (p23) -- (p24) -- (p25) -- (p23);
    \path[graph edge] (p25) -- (p27) -- (p22);
    \path[graph edge] (p21) -- (p27);
    \path[virtual edge] (p12) -- (p10);
    \path[virtual edge] (p24) -- (p21);
    \path[graph edge] (p34) .. controls (p35) and (p36) .. (p37);
    \path[virtual edge] (p34) .. controls (p38) and (p33) .. (p37);
    \path[virtual edge] (p39) -- (p40);
    \path[graph edge] (p41) -- (p40);
    \path[virtual edge] (p42) -- (p41);
    \path[graph edge] (p43) .. controls (p44) and (p45) .. (p46);
    \path[virtual edge] (p43) .. controls (p47) and (p48) .. (p46);
    \path[virtual edge] (p42) -- (p39);
    \path[graph edge] (p57) -- (p58) -- (p59);
    \path[virtual edge] (p57) -- (p59);
    \path[virtual edge] (p63) .. controls (p64) and (p65) .. (p66);
    \path[virtual edge] (p66) .. controls (p67) and (p68) .. (p63);
    \path[virtual edge] (p13) -- (p16);
    \path[virtual edge] (p37) .. controls (p72) and (p73) .. (p34);
    \path[virtual edge] (p43) .. controls (p74) and (p75) .. (p46);
    \path[virtual edge] (p63) .. controls (p113) and (p114) .. (p66);

    \node[tree node] (v106) at (p106) {S};
    \node[tree node] (v107) at (p107) {P};
    \node[tree node] (v1) at (p1) {S};
    \node[tree node] (v108) at (p108) {S};
    \node[tree node] (v109) at (p109) {P};
    \node[tree node] (v110) at (p110) {R};
    \node[tree node] (v111) at (p111) {P};
    \node[tree node] (v112) at (p112) {R};

    \path[tree edge] (v108) .. controls (p77) and (p78) .. (v109);
    \path[tree edge] (v111) .. controls (p81) and (p82) .. (v108);
    \path[tree edge] (v112) .. controls (p85) and (p86) .. (v111);
    \path[tree edge] (v109) .. controls (p89) and (p72) .. (p90)
                            .. controls (p91) and (p92) .. (v110);
    \path[tree edge] (v107) .. controls (p95) and (p96) .. (v108);
    \path[tree edge] (v107) .. controls (p99) and (p100) .. (v106);
    \path[tree edge] (v107) .. controls (p103) and (p104) .. (v1);

\end{scope}

\end{tikzpicture}

   \end{adjustbox}
  \caption{A biconnected graph and its SPQR tree.  See Definition~\ref{def:SPQR}.}
  \label{fig:SPQR}
\end{figure}
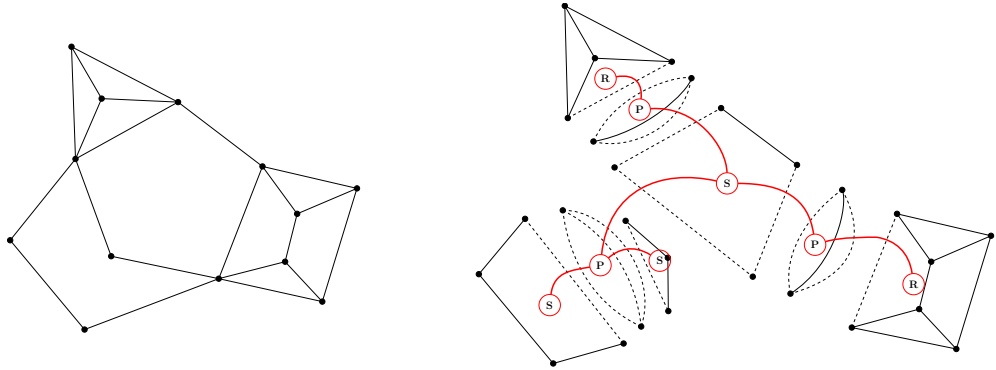
\begin{definition}\label{def:SPQR}
	The SPQR-tree for a biconnected multigraph $G=(V,E)$ with at least $3$ edges is a tree with nodes labeled S, P, or R, where each node $x$ has an associated \emph{skeleton graph} $\Gamma(x)$ with the following properties:
	\begin{itemize}
		\item For every node $x$ in the SPQR tree, $V(\Gamma(x))\subseteq V$.
		\item For every edge $(x,y)$ in the SPQR tree, $V(\Gamma(x))\cap V(\Gamma(y))$ is a separation pair $\set{a,b}$ in $G$, and~there is a virtual edge $ab$ in each of $\Gamma(x)$ and $\Gamma(y)$.
    \item For every node $x$ in the SPQR tree, every edge in $\Gamma(x)$ is either in $E$ or a \emph{virtual edge} corresponding %
    to an edge $(x,y)$ in the SPQR-tree.
		\item For every edge $e\in E$ there is a unique node $x$ in the SPQR-tree such that $e\in E(\Gamma(x))$.
		\item If $x$ is an S-node, $\Gamma(x)$ is a simple cycle with at least
		$3$ edges.
		\item If $x$ is a P-node, $\Gamma(x)$ consists of a pair of vertices with at least $3$ parallel edges.
    \item If $x$ is an R-node, $\Gamma(x)$ is a simple %
    triconnected graph.
		\item No two S-nodes are neighbors, and no two P-nodes are neighbors.
	\end{itemize}
\end{definition}
It turns out (see e.g.~\cite{Battista:96}) that the SPQR-tree for a biconnected graph is unique.  The (skeleton graphs associated with) nodes of the SPQR-tree are sometimes referred to as the triconnected \emph{components} of $G$.

\subparagraph*{Detecting separating $4$-cycles.}
A $4$-cycle is a simple cycle of length $4$.
We say that a $4$-cycle in a planar embedded graph $G$ is a \emph{face $4$-cycle} if it is a cycle bounding a face of $G$, and a \emph{separating $4$-cycle} otherwise.
As we show in Appendix~\ref{sec:SPQR-proofs}, there is a one-to-one correspondence between separation pairs in $G$ and separating $4$-cycles in the vertex-face graph $\fv{G}$.

Since no two parallel edges can lie on the same $4$-cycle, and no self-loop can be contained in a $4$-cycle, we can assume the input graph is simple. However, when we contract edges, new parallel edges and self-loops may arise.  To handle this, we could detect and remove parallel edges,
but it turns out that both the algorithm and the analysis become simpler if we keep (most of) the additional edges, as long as no two parallel edges are consecutive in the circular ordering around both their endpoints.

If $G$ has a face bounded by two edges, we can \emph{simplify} the graph by
deleting one of them. For our purposes we do not really care which one is deleted, but we need a rule that is consistent. For presentational purposes,
we assume that we always keep the edge $e$ with larger $\eid(e)$.

This motivates the following definition of a quasi-simple graph\footnote{In~\cite{Klein:book} these graphs are called \emph{semi-strict}.}:
\begin{definition}\label{def:quasisimple}
	A plane embedded graph is \emph{quasi-simple} if the dual of each
	non-simple component has minimum degree $3$.
    Given a plane embedded graph $G$ and a set of vertices $X$, we define the
	subgraph of $G$ \emph{quasi-induced} by $X$ to be the unique
  quasi-simple subgraph of $G$ with vertex set $X$ and the maximum total sum of $\eid(e)$ values.  Let $d_X(v)$
	denote the degree of $v$ in the subgraph quasi-induced by
	$X\cup\{v\}$.
\end{definition}
Roughly speaking, a quasi-simple graph is obtained from a plane embedded multigraph by merging parallel edges that lie next to each other in the circular orderings around both their endpoints.

We build a structure for $4$-cycle detection by recursively using balanced separators, and by detecting, for each separator, the cycles that cross the separator. Detecting $4$-cycles that cross a separator is not trivial, and our analysis introduces a complicated potential function which reflects how well connected the non-separator vertices are with the separator, that is, how many neighbors on the separator they have. At the
same time, we make sure that all the work done can be paid with the decrease in the potential.
Our analysis exploits the fact that for a planar graph with separator $S$, at most $O(|S|)$ vertices have more than $4$ neighbors in $S$. 

The recursive use of separators can be sketched as follows: 
Let $S$ be a small balanced separator in $G=(V,E)$ that induces a separation $(V_1, V_2)$, that is, $V_1\cap V_2=S$ and $V_1\cup V_2 = V$. Moreover, let $n = |V|$.
We observe that each $4$-cycle is fully contained in $V_1$ or $V_2$, or consists of two paths of length $2$ that connect vertices of $S$.
This motivates the following recursive approach.
We compute a separator $S$ of $O(\sqrt{n})$ vertices and then find all paths of length $2$ that connect vertices of $S$.
Since the size of $S$ is $O(\sqrt{n})$, there are only $O(n)$ pairs of vertices of $S$, and for each pair of vertices, we can easily check if the two-edge paths connecting them form any separating $4$-cycles.
It then remains to find the $4$-cycles that are fully contained in either $V_1$ or $V_2$, which can be done recursively.
Because $S$ is a balanced separator, the recursion has $O(\log n)$ levels.

This algorithm can be made dynamic under contractions and edge insertions that respect the embedding of $G$.
Contractions are easy to handle, as they preserve planarity.
Moreover, a separator $S$ of a planar graph can be easily updated under contractions.
Namely, whenever an edge $uw$ is contracted, the resulting vertex belongs to the separator iff any of $u$ and $w$ did.
Insertions that preserve planarity, however, are in general harder to accommodate. To handle this we introduce a new type of separators that we call \emph{face-preserving} separators, which (like cycle-separators) always exist when the face-degree is bounded. These are still preserved by contractions, but also ensure that any edge across a face can be inserted.

All in all, there are $O(\log n)$ levels of size $O(n)$ each, where each level handles insertions and contractions in constant time, leading to a total of $O(n \log n)$ time.
The details of this construction and the proof of the below theorem can be found in Appendix~\ref{sec:4cycle}.
\begin{restatable}{theorem}{cycledetection}
\label{thm:4cycledetection}
  Let $G$ be an $n$-vertex quasi-simple plane embedded graph with bounded face degree.
  There exists a data structure that maintains $G$ under contractions and embedding-respecting insertions, 
  and after each update operation reports edges that become members of some separating $4$-cycle.
  It runs in $O(n \log n)$ total time.
\end{restatable}

\subparagraph*{Maintaining SPQR trees.}
The main challenge in maintaining an SPQR-tree is handling the case when an edge within a triconnected component is deleted. First of all, the data structure should be able to detect whether or not the component is still triconnected. 

For any triconnected component $\Gamma$ of $G$, we maintain a $4$-cycle detection structure for the corresponding vertex-face graph $\fv{\Gamma}$. A separating $4$-cycle in $\fv{\Gamma}$ corresponds to a separation pair in $\Gamma$, which would witness that $\Gamma$ is no longer triconnected. The deletion or contraction of the edge $e$ in the triconnected component $\Gamma$ of $G$ corresponds to an (embedding-respecting) insertion and immediate contraction of an edge in $\fv{\Gamma}$. This way by detecting $4$-cycles in $\fv{\Gamma}$, we can detect when the corresponding triconnected component falls apart.

However, this is not the only challenge. If $\Gamma$ does indeed cease to be triconnected, the SPQR-tree of $(\Gamma-e)$ (or $(\Gamma/e)$ when doing a contraction) is a path $\mathcal{P}$. This is where we need the $4$-cycle structure to output the edges contained in separating $4$-cycles. Those edges correspond to a set of corners $N$ of $G$. We use those corners to guide a search, which helps identify the non-largest components of the SPQR-path $\mathcal{P}$.
More specifically, if a vertex $v$ now belongs to two distinct triconnected components, there are two corners in $N$ that separate the edges incident to $v$ into two groups of edges, each belonging to a distinct triconnected component.
We can afford to build a $4$-cycle detection structure for $\fv{\Gamma'}$ for any non-largest triconnected component $\Gamma'$ on the path from scratch.
To obtain the data structure representing the largest component, we delete or contract the corresponding edges from $\Gamma$ while updating $\fv{\Gamma}$.
Since an edge only becomes part of a structure built from scratch when its triconnected component size has been halved, this happens only $O(\log n)$ times per edge, so the total time used for rebuilding is $O(n\log^2 n)$.
The second logarithmic factor comes from rebuilding the data structure for $4$-cycle detection, that takes $O(n \log n)$ time to initialize and process any number of operations.

Finally, since no two $S$-nodes can be neighbors and no two $P$-nodes can be neighbors, some $S$- or $P$-nodes in $\mathcal{P}$ may have to be merged with their (at most $2$) neighbors of the same type outside $\mathcal{P}$. To handle this step efficiently, we keep the SPQR-tree rooted in an arbitrary node.
While merging the skeleton graphs of two $S$- or $P$-nodes can be done in constant time, what can be more costly is updating the parent pointers in the children of the merged nodes.
Hence, we move the children of the node with fewer children to the other node. This way, each node changes parent at most $O(\log n)$ times before it is deleted or split. The total number of distinct SPQR-nodes that exist throughout the lifetime of the data structure is $O(n)$, so the total time used for maintaining the parent pointers is $O(n\log n)$.

Since SPQR-trees are only defined for biconnected graphs, another challenge is to maintain SPQR-trees for each biconnected component, even as the decremental update operations cause the biconnected components to fall apart.  
We recall here that
the structure of the biconnected components of a connected graph can be described by a tree called the \emph{block-cutpoint tree}~\cite[p. 36]{Harary69}, or \emph{BC-tree} for short.
This tree has a vertex for each biconnected component (block) and for each articulation point of the graph, and an edge
for each pair of a block and an articulation point that belongs to that block. If the tree is rooted arbitrarily at any block, each non-root block has a unique articulation point separating it from its parent.

To handle updates, we notice that the SPQR-tree points to the fragile places where the graph is about to cease to be biconnected: An edge deletion in an $S$-node will break up a block in the BC-tree into path, and an edge contraction in a $P$-node breaks a block in the BC-tree into a star. 
Upon such an update, we remove the aforementioned $S$- or $P$-node from the SPQR-tree, breaking it up into an SPQR-forest. Each tree corresponds to a new block in the BC-tree. They form a path (or a star), and the ordering along the path, as well as the articulation points, can be read directly from the SPQR-tree.

On the other hand, in order to even know which SPQR tree to modify during an update, we can search in the BC-tree for the right SPQR-structure in which to perform the operation.

\subparagraph*{Bi- and triconnectivity.}
Finally, we use SPQR-trees to facilitate triconnectivity queries. 
First of all, vertices need to be biconnected in order to be triconnected. 
To facilitate biconnectivity queries, it is enough that 
each vertex $v$ knows the name of the block $B(v)$ closest to the root in the BC-tree that contains it, and each block $b$ knows the name of the vertex separating it from the parent $p(b)$. Then, any two vertices $u$ and $w$ are biconnected if and only if one of the following occur: $B(u)=B(v)$, or $u=p(B(v))$, or $v=p(B(u))$.

The information we maintain for triconnectivity is similar, using the SPQR-tree.
Namely:
each non-root node $x$ in the SPQR-tree stores the \emph{virtual edge} (see Definition~\ref{def:SPQR}) that separates it from its parent.
Each vertex $v$ stores (a pointer to) the node $C(v)$ closest to the root that contains it, and, in a special case, at most two other nodes that are the children of $C(v)$. Queries are handled similarly as above.

The main challenge is to handle updates. Note that the change to the SPQR-tree may involve both the split and merge of nodes. In particular, we have one split and up to several merges when a triconnected component falls apart into an SPQR-path. However, upon a merge, we can afford to update the information regarding vertices in the non-largest components, costing only an additive $\log n$ to the amortized running time. Similarly, upon a split, we update any information that relates to vertices in the non-largest components only.

The total running time is thus $O(n\log n +f(n))$, where $f(n)$ is the running time for maintaining the SPQR-tree. The following theorem is proven in Appendix~\ref{sec:3vertex}.
\begin{restatable}{theorem}{threevertex}\label{thm:3vertex}
	There is a data structure that can be initialized on a planar graph $G$ on $n$ vertices in $O(n \log n)$ time, and supports any sequence of $k$ edge deletions or contractions in total time $O((n+k)\log^2 n)$, while supporting queries to pairwise triconnectivity in worst-case constant time per query.
\end{restatable}

\newcommand{\parf}{\ensuremath{\mathcal{T}}}

\section{Decremental SPQR-trees}\label{sec:SPQR}
In this section, we %
use the data structure of Theorem~\ref{thm:4cycledetection} (described in Appendix~\ref{sec:4cycle}) to maintain an SPQR-tree (see Definition~\ref{def:SPQR}) for each biconnected component of $G$ with at least $3$ edges under arbitrary edge deletions and contractions.
We start with some useful facts.

\begin{restatable}{lemma}{fvCycleEdgeclasses}\label{lem:fvCycleEdgeclasses}
    For any pair of edges in a biconnected graph $G$, their corresponding faces of $\fv{G}$ are separated by a $4$-cycle $(v_1,f_1,v_2,f_2)$ if and only if they belong to different separation classes with respect to $v_1,v_2$ in $G$ and with respect to $f_1,f_2$ in $\dual{G}$.
\end{restatable}

\begin{restatable}{lemma}{sepone}\label{lem:sep1}
  Let $G$ be a biconnected graph.
    If a $4$-cycle $C=(v_1,f_1,v_2,f_2)$ in $\fv{G}$ is a separating cycle, %
    then $v_1,v_2$ is a separation pair of $G$ and $f_1,f_2$ is a separation pair of $\dual{G}$.
\end{restatable}

\begin{restatable}{lemma}{classescircular}\label{lem:classescircular}
Let $G$ be a loopless biconnected plane graph and $u$, $w$ be a separation pair in $G$.
Consider the set of edges $E_x$ incident to $x \in \{u, w\}$.
Then, the edges of $E_x$ belonging to each separation class of $u, w$ are consecutive in the circular ordering around both $u$ and $w$.
\end{restatable}

\begin{restatable}{lemma}{spqrpath}\label{lem:spqrpath}
Let $G$ be a triconnected plane graph and $e=uw \in E(G)$.
Assume that $G - e$ is not triconnected.
Then, the SPQR-tree of $G - e$ is a path $H$ (we call it an SPQR-path).
Moreover, given all edges that lie on $4$-cycles in $\fv{(G - e)}$, we can compute all nodes of $H$ except for the largest one in time that is linear in their size.
\end{restatable}

For a planar graph, there is a nice duality, as proven by Angelini et al.~\cite[Lemma 1]{Angelini2013}. Define the dual SPQR-tree as the tree  obtained from the SPQR-tree by interchanging $S$- and $P$-nodes, and taking the dual of the skeletons. 
\begin{lemma}[Angelini et al~\cite{Angelini2013}]\label{lem:dualSPQR}
  The SPQR-tree of $\dual{G}$ is the dual SPQR-tree of $G$.
\end{lemma}

Let $G$ be a connected plane graph. Since $\dual{(\fv{G})}$ is $4$-regular, $\fv{G}$ is quasi-simple and has bounded face-degree.  Furthermore, any edge deletion or contraction in $G$ that leaves $G$ connected, corresponds to an  edge insertion and immediate contraction in $\fv{G}$.  Thus by Theorem~\ref{thm:4cycledetection} we can maintain a data structure for $G$ under connectivity-preserving edge deletions and contractions, that after each update operation reports the corners that become part of a separating $4$-cycle in $\fv{G}$.

\begin{algorithm}[H]
	\begin{algorithmic}[1]
		\Function{\removep}{$e, x, T$}
		\State remove $e$ from $\Gamma(x)$
		\If{$\Gamma(x)$ has two edges}
		\If{$\Gamma(x)$ has no virtual edges}
		\State{delete $T$}
		\ElsIf{$\Gamma(x)$ has one virtual edge}
		\State{$y := $ the only neighbor of $x$}
		\State{$e_x := $
			\parbox[t]{\columnwidth-4cm}{
				the virtual edge in $\Gamma(y)$ corresponding to $x$
			}
		}
		\State{replace
			\parbox[t]{\columnwidth-4cm}{
				$e_x$ by the non-virtual edge of $\Gamma(x)$
			}
		}
		\State{remove $x$ from $T$}
		\ElsIf{$\Gamma(x)$ has two virtual edges}
		\State{$\{y, z\}$ := neighbors of $x$ in $T$}
		\State{remove $x$ from $T$,
			\parbox[t]{\columnwidth-4cm}{
				making $y$ and $z$ neighbors in $T$
			}
		}
		\If{$y$ and $z$ are $S$-nodes}
		\State{merge $y$ and $z$ into one node}
		\EndIf
		\EndIf
		\EndIf
		\EndFunction
	\end{algorithmic}
	\caption{\label{alg:removep}Removing an edge $e$ from a $P$-node $x$ of $T$}
\end{algorithm}

In the algorithm we maintain one SPQR-tree for each biconnected component with at least $3$ edges.
We now describe how these trees are updated upon edge deletions.
The procedures, depending on the type of the SPQR-tree node are given as
Algorithms~\ref{alg:removep}, \ref{alg:remover} and~\ref{alg:removes}.
Note that the lines~\ref{l:tree} and~\ref{l:big} in Algorithm~\ref{alg:remover} only introduce notation, that is the values of the variables are not computed.
The proofs of %
correctness can be found in Appendix~\ref{sec:SPQR-proofs}.

\begin{multicols}{2}
\newcommand{\pushcode}[1][1]{\hskip\dimexpr#1\algorithmicindent\relax}
\begin{algorithm}[H]
\begin{algorithmic}[1]
\Function{\remover}{$e, x, T$}
\State remove $e$ from $\Gamma(x)$
\If{$\Gamma(x)$ has a separation pair}
	\State{\label{l:tree}$X' := $ 
          \parbox[t]{\columnwidth-3cm}{
			SPQR-path representing $\Gamma(x)$}
		}
	\State{\label{l:big}$x_{big} := $ 
          \parbox[t]{\columnwidth-3cm}{
            the node of $X'$ 
            st. $\Gamma(x_{big})$ has the most edges
          }
        }
	\State{compute all nodes of $X'\setminus x_{big}$}
	\State{remove 
          \parbox[t]{\columnwidth-3cm}{
            and contract edges of $\Gamma(x)$ to obtain $\Gamma(x_{big})$
          }
        }
	\State{replace
          \parbox[t]{\columnwidth-3cm}{$x$ in $T$ by $X'$ (connect each child of $x$ to the correct node of $X'$)
          }
        }
	\For{each $S$- or $P$-node $z \in X'$}
        	\For{each neighbor $z'\notin X'$} 
			\If{$z,z'$ are same type}%
				\State{merge $z$ with $z'$}
			\EndIf
		\EndFor
	\EndFor
\EndIf
\EndFunction
\end{algorithmic}
	\caption{\label{alg:remover}Removing an edge $e$ from an R-node $x$ of $T$}
\end{algorithm}

\columnbreak

\begin{algorithm}[H]
\begin{algorithmic}[1]
\Function{\removes}{$e, x, T$}
\State remove $e$ from $\Gamma(x)$
\State remove $x$ from $T$
\For{each edge $e'$ in $\Gamma(x)$}
        \State Make a new BC-node $z$
	\If{$e'$ is a virtual edge}
		\State{$y := $
                  \parbox[t]{\columnwidth-3cm}{
                    neighbor of $x$ in $T$ corresponding to $e'$
                  }
                }
        	\State{\parbox[t]{\columnwidth-2cm}{Make 
                    the tree containing $y$ the SPQR-tree for the new BC-node
                  }
                }
		\If{$y$ is a $P$-node}
			\State $\removep(y, e', T)$
		\Else
			\State $\remover(y, e', T)$
		\EndIf
        \EndIf
\EndFor
\EndFunction
\end{algorithmic}
	\caption{\label{alg:removes}Removing an edge $e$ from an $S$-node $x$ of $T$}
\end{algorithm}
  
\end{multicols}

We can now prove the main theorem of this section.
Note that, as in the block-cutpoint tree, we root each SPQR-tree in an arbitrary vertex.

\begin{theorem}\label{thm:SPQR}
There is a data structure that can be initialized on a simple planar graph $G$ on $n$ vertices in $O(n \log n)$ time, and supports any sequence of edge deletions or contractions in total time $O(n \log^2 n)$, while maintaining an explicit representation of a rooted SPQR-tree for each biconnected component with at least $3$ edges, including all the skeleton graphs for the triconnected components.
Moreover, %
during updates, the total number of times a node of an SPQR-tree changes its parent is $O(n \log n)$.
\end{theorem}

\begin{proof}
We first partition the graph into biconnected components, 
and, as sketched in Section~\ref{sec:overview}, maintain the block-cutpoint tree explicitly.
Thus, given two vertices $u,v$, we can in $O(1)$ time access the biconnected component containing both of them, along with its auxiliary data.  Now, for each biconnected component $C_i$, we compute the SPQR-tree $T$.  This can be done in linear time due to~\cite{Gutwenger2001}.
We also root each SPQR-tree in an arbitrary node, and keep the trees rooted as they are updated.

For each node $x$ of $T$ we maintain the graph $\Gamma(x)$.
Each virtual edge of $\Gamma(x)$ has a pointer to the neighbor of $x$ it represents.
  Moreover, for each R-node $r$, we keep a data structure of Theorem~\ref{thm:4cycledetection} for detecting separating $4$-cycles in the vertex-face graph $\fv{(\Gamma(r))}$.
  By Lemma~\ref{lem:sep1}, any separating $4$-cycle in $\fv{(\Gamma(r))}$ corresponds to a separation pair in $\Gamma(r)$.
Since $r$ is an $R$-node, there are no separating $4$-cycles to begin with, but some may appear after an update.

Since the total size of the $R$-components is $n$, it follows from Theorem~\ref{thm:4cycledetection} that the entire construction time is $O(n \log n)$.

\subparagraph*{Deletion.}
When an edge $e$ is removed we find the node $x$ of the SPQR-tree, such that $e$ is a non-virtual edge in $x$.
Then, we proceed according to Algorithms~\ref{alg:removep}, \ref{alg:remover} and \ref{alg:removes}.

  Whenever an edge $fg$ is deleted from an R-node $r$, we update the corresponding $4$-cycle detection structure for $\fv{(\Gamma(r))}$.  We first insert the dual edge $\dual{(fg)}$ in the vertex-face graph, and then contract along that edge.
This allows us to detect whether $\Gamma(r)$ has any separation pairs after each edge deletion.

Let us now analyze the running time. When processing an edge deletion, the following changes can take place in a SPQR-tree (all other changes can be handled in $O(1)$ time):
\begin{itemize}
\item an R-node is split into multiple nodes,
\item two $P$-nodes or $S$-nodes are merged,
\item an $S$- or $P$- node is deleted.
\end{itemize}
Note, a $P$- or $S$-node can never get split.
So, though each edge may at first belong to nodes that are split, once it becomes a part of a $P$- or $S$-node, its node only participates in merges.

When two $S$- or $P$-nodes are merged, we can merge their skeleton graphs in constant time.
These skeleton graphs have only two common nodes, and their lists of adjacent edges can be merged in constant time thanks to Lemma~\ref{lem:classescircular}.
When nodes are merged, we also have to update the parent pointers of their children.
To bound the number of these updates, we merge the node with fewer children into the node with more.
Thus, the number of parent updates caused by these merges is $O(n \log n)$, and so is the impact on the running time.

A similar analysis applies to the case when an R-node $r$ is split into an SPQR-path.
By Lemma~\ref{lem:spqrpath}, we can compute all but the largest node of the SPQR-path in linear time.
Since the size of the skeleton graph in each of these nodes is at most half the size of $\Gamma(r)$, each edge takes part in this computation at most $O(\log n)$ times.
For every new R-nodes, we also initialize their associated data structures for detecting $4$-cycles.
We charge the running time of each data structure to this initialization.
From Theorem~\ref{thm:4cycledetection} we get that recomputing all the nodes and data structures takes $O(n \log^2 n)$ total time.

Taking care of the largest component of the SPQR-path is even easier, as we can simply reuse the skeleton graph of $r$ and its associated data structure for detecting $4$-cycles.
To update the skeleton graph, we use the following lemma.

\begin{restatable}{lemma}{DeleteContractSeq}\label{lem:DeleteContractSeq}
  If $G$ is triconnected, $e\in E(G)$, and $x$ is an $R$-node in the SPQR-tree for $G-e$, then there exists a sequence of $\sizeof{E(G)\strut}-\sizeof{E(\Gamma(x))\strut}$ edge deletions and contractions that transform $G-e$ into $\Gamma(x)$ while keeping the graph connected at all times.
\end{restatable}

After an R-node $r$ is split into a SPQR-path $H$ we also need to update the parent pointers in the children of $r$.
However, the number of children to update is at most the number of edges in the non-largest components of the SPQR-path.
As we have argued, the total number of such edges across all deletions is $O(n \log n)$.

\subparagraph*{Contraction.}
    The contraction of an edge of the embedded planar graph $G$ corresponds to the deletion of an edge of its dual graph, $\dual{G}$. By Lemma~\ref{lem:dualSPQR}, the SPQR-tree of $\dual{G}$ is the dual SPQR-tree of $G$. %
    Thus, if the edge was in %
    a $P$-node of the SPQR-tree, its contraction is handled like the deletion of an edge in a $S$-node, and vice versa. 
    
    If the contracted edge $e$ belongs to an $R$-node, that $R$ node may expand to a path in the SPQR-tree (because deletion in $\dual{G}$ may expand an $R$-node into a path). In the vertex-face graph, we may find all edges participating in new separating $4$-cycles, corresponding to separating corners of the graph.
  To find the new components, we simply apply Lemma~\ref{lem:spqrpath} to the dual graph and proceed analogously to a deletion.
\end{proof}

 \newpage
\appendix
\section{Decremental Triconnectivity}\label{sec:3vertex}

To answer triconnectivity queries, we maintain a rooted SPQR-decomposition
(see e.g.~\cite{Battista:96, Gutwenger2001}) of each biconnected
component of the planar graph.

Now it follows from the definition that pair of vertices in a biconnected graph is triconnected if and only if there exists a $P$ or $R$ component in the SPQR-tree containing them both.
By associating a constant amount of information with every vertex in $G$ and every node in the SPQR-tree, we can answer triconnectivity queries in constant time:
\begin{definition}\label{def:pointers}
  A \emph{triconnectivity query structure} for a biconnected graph consists of a rooted SPQR-tree, and the following additional information:
\begin{itemize}
\item For each node $x$ in the SPQR-tree except the root, a pointer $e(x)$ to the virtual edge that separates it from its parent.
\item For each vertex $v$, a pointer $C(v)$ to the node containing $v$ that is closest to the root.
\item For each vertex $v$ such that $C(v)$ points to an $S$-node $x$, a set $D(v)$ of pointers to the at most $2$ children of $x$ that contain $v$.
\end{itemize}
\end{definition}

\begin{lemma}\label{lem:3query}
Given the triconnectivity query structure described in Definition~\ref{def:pointers}, we can answer triconnectivity for any pair of vertices in constant time.
\end{lemma}

\begin{proof}
  Given vertices $u$ and $v$. If $C(u)=C(v)$ and $C(u)$ is not an S-node, then $u$ and $v$ are triconnected. %
If $C(u)=C(v)$ is an S-node, $u$ and $v$ are triconnected if and only if they share a virtual edge in $\Gamma(C(u))$, which happens if and only if $D(u)\cap D(v)$ is non-empty.
If $C(u)\neq C(v)$, then $u$ and $v$ are triconnected if and only if either $u$ is an endpoint of $e(C(v))$, or, $v$ is an endpoint of $e(C(u))$.
\qedhere
\end{proof}

Given Theorem~\ref{thm:4cycledetection}, we have the tools ready for maintaining triconnectivity: 

\threevertex*
\begin{proof}
  For each vertex $v$ and for each SPQR-node $x$, we associate the information $e(x),C(v),D(v)$ described in Definition~\ref{def:pointers}.

\subparagraph*{Query.}
  To answer a triconnected query $(u,v)$, we first ask if $(u,v)$ are biconnected.  Otherwise, they cannot be triconnected. If they are, we get the SPQR-tree associated with their common biconnected component and use Lemma~\ref{lem:3query}. This answers the query in $O(1)$ worst case time.

\subparagraph*{Updates.}
  Our data structure for SPQR trees already maintain $e(x)$, so the main difficulty is in maintaining $C(v)$ and $D(v)$ for each vertex.  Let $x$ be the value of $C(v)$ before the change, let $x'$ the new value, and suppose $x\neq x'$.

If $x$ and $x'$ are both $R$-nodes, $\abs{E(\Gamma(x'))}<\frac{1}{2}\abs{E(\Gamma(x))}$ so we are already using $\Omega(\abs{E(\Gamma(x'))})=\Omega(\abs{V(\Gamma(x'))})$ time to rebuild $\fv{\Gamma(x')}$.  We can thus afford to update $C(v)$ for all $v\in V(\Gamma(x'))$.

If $x$ is an $R$-node and $x'$ is not, then $C(c)$ was split into $k>1$ new nodes. In this case we are already using $\Omega(k)$ time maintaining the SPQR tree, so we can %
spend an additional $O(k)$ time on updating $C(v)$ for the 
$O(k)$ vertices from $V(\Gamma(x))$ whose new $C(v)$ is not an $R$-node.

If $x$ is a $P$-node, then it has to be the root (since $C(v)=x$), so this can happen for at most $2$ vertices per update and we can easily afford that.

If $x$ is an $S$-node, then either the biconnected component was split into $k>1$ new components and we can afford to spend $O(k)$ time on updating $C(v)$ for the $k-2$ vertices in $S$ that were pointing to $x$.  Or $x$ was merged into another $S$-node.  The total cost is linear in the total number of times some node changes parent due to such a merge, which is $O(n\log n)$.

Finally for each node $x'$ that has a new parent $p$ in the SPQR-tree, if $p$ is an $S$-node, $e(x')$ has the two vertices whose $D(v)$ need to be changed, and this can be done in constant time.  The total number of times this happens is $O(n\log n)$.
\end{proof}

\section{Detecting $4$-Cycles Under Edge Contractions and Insertions}\label{sec:4cycle}

In this section we give an algorithm for detecting $4$-cycles (simple cycles of length $4$) in a planar embedded graph that undergoes contractions and edge insertions that respect the embedding.
We say that a $4$-cycle in a planar embedded graph $G$ is a \emph{face $4$-cycle} if it is a cycle bounding a face of $G$, and a \emph{separating $4$-cycle} otherwise. For our purposes, only the separating $4$-cycles are interesting, but we note in passing that new face $4$-cycles are easy to detect under edge insertions and contractions:
\begin{observation}
	An embedding-respecting edge insertion creates two new faces, and we may check in constant time whether each of them has degree $4$ or not.
	An edge contraction affects degrees of only two faces (the two incident to the contracted edge), and we may check in constant time whether their new degree is $4$ or not.
\end{observation}

The main goal of this section is to prove the following theorem.
\cycledetection*

In order to detect $4$-cycles, we use planar separators.  In fact, in order to maintain our data structure dynamically, we need something a little bit stronger.
\begin{definition}
	Given a planar graph $G$, a separation $(A,B)$ of $G$ is said to be \emph{face-preserving} if for any face $f$ of $G$, all vertices of $f$ belong to $A$ or all vertices of $f$ belong to $B$.
\end{definition}
For instance, given a cycle separator $K$, we can form a face-preserving separation $(A, B)$ such that $A \cap B = K$. Namely, $K$ corresponds to a Jordan curve dividing the plane into two parts, $S_A,S_B$, where every face lies entirely in one part. Define $A$ by all vertices incident to faces on $S_A$, and $B$ similarly. Then, $A \cup B = G$, and $A\cap B = K$.

In our algorithm we need to maintain separations under edge insertions and contractions.
Let $(A, B)$ be a separation in $G$.
When an edge is inserted, we do not modify the separation.
When an edge $uw$ is contracted into a vertex $x$, we obtain a new separation $(A', B')$ as follows.
If $u \in A$ or $w \in A$, we set $A' = (A \setminus \{u, w\}) \cup \{x\}$.
Otherwise, $A' = A$.
The set $B'$ is defined analogously.
Thanks to this convention, we obtain the following.

\begin{lemma}
  Let $(A,B)$ be a face-preserving separation in $G$. Let $G'$ be the result of an embedding-respecting edge insertion or edge contraction, and let $A',B'$ be the vertices corresponding to $A$ and $B$ in $G'$.  Then $(A',B')$ is a face-preserving separation in $G'$, and $\sizeof{A\cap B}-1\leq \sizeof{A'\cap B'} \leq \sizeof{A\cap B}$.
\end{lemma}
\begin{proof}
  If an edge is inserted that respects the embedding, it is inserted into some face $f$. By definition at least one of $A,B$ contain all vertices on $f$, and in particular it also contains all the vertices of the two new faces that appear in $G'$. Since $A=A'$ and $B=B'$ in this case, the result~follows.

  If an edge $uv$ is contracted, the resulting graph $G'$ has the same faces as $G$, and the separation $(A',B')$ is clearly face-preserving.
  If $u,v\in A\cap B$, then $\sizeof{A'\cap B'}=\sizeof{A\cap B}-1$. Otherwise $uv$ has an endpoint outside $A\cap B$. Without loss of generality, we can assume that $u\in A\setminus B$. In that case, $v\in A$ and $B=B'$ and it follows that $\sizeof{A'\cap B'}=\sizeof{A\cap B}$.
\end{proof}

\begin{definition}\label{def:separator-tree}
  Given a graph $G$, a \emph{separator tree} is a binary tree where each node $x$ is associated with an induced subgraph $H_x$ of $G$, such that for some constant $n_0>0$:
  \begin{itemize}
  \setlength{\parskip}{0pt}
  \setlength\itemsep{0.5em}
  \item If $x$ is the root, $H_x=G$.
  \item If $\abs{V(H_x)}>n_0$ then $x$ has children $y$, $z$ such that
    $(V(H_y),V(H_z))$ is a balanced separation of $H_x$ with a small
    separator $S_x=V(H_y)\cap V(H_z)$.
  \item If $\abs{V(H_x)}\leq n_0$ then $x$ is a leaf.
  \end{itemize}
  A separator tree is a \emph{cycle separator tree} if $S_x$ is a cycle separator, and it is \emph{face preserving} if $(V(H_y),V(H_z))$ is face-preserving, for all nodes $x$ with children $y$ and~$z$.
\end{definition}

\begin{lemma}
  \label{lem:septree}
  Given a planar graph with bounded face degree, we can in $O(n\log n)$ time build a face-preserving separator tree where each node $x$ explicitly stores $S_x$ and $H_x$. This tree has height $O(\log n)$, and uses $O(n\log n)$ space.
\end{lemma}
    We construct the tree in three steps. First, we take our graph $G$ and make a triangulation $G^\triangle$. Then, referring to a result by Klein, Mozes, and Sommer~\cite{Klein:13}, we make a cycle separator tree for $G^\triangle$. %
    Finally, we can transform the cycle separator tree for $G^\triangle$ to a face preserving separator tree for $G$:
\begin{proof}
  Let $G$ be a graph with maximum face-degree $k$ and let $G^\triangle$ be a triangulation of $G$.  Then using the algorithm from~\cite[Theorem 3]{Klein:13}, we can in linear time compute a cycle separator tree for $G^\triangle$. Since the cycle separator tree is balanced, it has height $h\in O(\log n)$.  Since the children of each node partition the faces of $G^\triangle$ contained in the node, and contain at most a constant number of other faces (called ``holes''), the total number of faces (and hence vertices) in graphs associated with depth $i$ nodes is $O(n)$ for each $0\leq i<h$. %
  Thus the total size of all these graphs in the cycle separator tree is $O(n\cdot h)=O(n\log n)$.
  
  While this cycle separator tree indeed is face preserving for $G^\triangle$, it may be not face preserving for $G$ because we may have edges $e\in G^\triangle \setminus G$ in the cycle separator. Luckily, for each such edge, we can fix the problem by adding the at most $k$ vertices of the crossed face to the separator as follows:
  
For each $H^\triangle_x$ in the cycle separator tree for $G^\triangle$, we can construct $H_x$ by starting with $V(H^\triangle_x)$, adding the remaining vertices of each face of $G$ that is crossed by an edge in $E(H^\triangle_x)\setminus E(G)$, and taking the induced subgraph of $G$.  By construction, $\sizeof{H^\triangle_x}\leq \sizeof{H_x}\leq O(k\sizeof{H^\triangle_x})$. Furthermore, if $x$ is parent to $y,z$ in the separator tree then $V(H_x)=V(H_y)\cup V(H_z)$.
Clearly, $(V(H_y),V(H_z))$ is a balanced separation of $H_x$, 
and, by construction, it is face-preserving. What remains to be shown is that the size of the corresponding separator is small.

Given a node $x$ in the separator tree with children $y,z$ we
compute $S_x = V(H_y)\cap V(H_z)$.  Note that each vertex in
$S_x\setminus S^\triangle_x$ must belong to a face of $G$ that is crossed by an edge in $S^\triangle_x$, so also $\sizeof{S^\triangle_x}\leq\sizeof{S_x}\leq O(k\sizeof{S^\triangle_x})$.
Thus the $H_x$ and $S_x$ form a face-preserving separator tree for $G$.

The total time to explicitly construct the face-preserving separator tree and all the associated graphs is $O(k n \log n)$.
\end{proof}

\begin{lemma}\label{lem:4cycles-crossing-separator}
	Let $G$ be a graph and $(A,B)$ be a separation of $G$.
	Then, any $4$-cycle either has exactly one vertex in $A\setminus B$, one vertex in $B\setminus A$, and the remaining two %
	vertices in $A\cap B$, or the $4$-cycle is completely contained in at least one of $A$ or $B$.
\end{lemma}

\begin{proof}
From Definition~\ref{def:separation}, for each edge $e$ of $G$, both endpoints of $e$ are in $A$ or $B$.
Thus, an edge that has one endpoint in $A \setminus B$ has its other endpoint in $A$.
Using these facts, the lemma follows by simple case analysis.
\end{proof}

It follows that we can use the separator tree to detect $4$-cycles as follows.
For each leaf $x$ of the separator tree, the graph $H_x$ has constant size, so $4$-cycles inside $H_x$ can be detected in constant total time.
In any other node, we have a graph with a separator $K$, and we need to dynamically detect $4$-cycles that cross $K$ under edge contractions and embedding-respecting edge insertions. Referring to Lemma~\ref{lem:4cycles-crossing-separator} above, we only need to detect the two halves of a $4$-cycle, that is, length-$2$ paths between the vertices of $K$.

\begin{lemma}\label{lem:ineachnode}
    Let $G$ be a plane embedded graph on $n$ vertices and $K$ be a set of vertices of $G$ of size $|K| = O(\sqrt{n})$.
    Assume that $G$ undergoes edge contractions and insertions respecting the embedding.
    There exists a data structure that after each update operation can report the edges of $G$ that become members of separating $4$-cycles whose 
	two opposing vertices lie on $K$.
	Its total running time is $O(n+k)$, where $k$ is the total number of edge contractions and insertions.
\end{lemma}

We now proceed with the description of the data structure of the above lemma.
First note that if a pair of vertices is connected by at least $4$ paths of length two, all edges on those paths lie on separating $4$-cycles (see Figure~\ref{fig:nontriv}). Thus, if we can keep, for every pair of vertices of $K$, the list of all length $2$ paths between them, we need to check at most $2$ existing paths when a new path arrives, and then report at most $8$ edges ($4$ new length-$2$ paths) that now belong to separating $4$-cycles.

To report every edge only once, we also keep a Boolean flag for each edge that indicates whether it has been reported before, and check that flag before reporting.

We thus only need to argue that we can detect all the length-$2$ paths between $K$-vertices that arise in the graph under contractions and edge insertions, in $O(n+k)$ total time. We do that by constructing a potential function $\Phi$ that is initially $O(n)$, remains nonnegative, and drops at each operation proportionally to the amount of work done.  We start by partitioning the vertices into $3$ sets, that we need to treat differently.

\begin{figure}
  \center
  \begin{minipage}[b]{.4\textwidth}
    \center
    \begin{adjustbox}{width=.7\textwidth}
\begin{tikzpicture}[y=0.80pt, x=0.80pt, yscale=-1.000000, xscale=1.000000,
inner sep=0pt,outer sep=0pt]
  \begin{scope}[
      edge/.style={
        draw=black,
        line join=miter,
        line cap=butt,
        miter limit=4.00,
        even odd rule,
      },
      vertex/.style={
        fill=black,
        line join=miter,
        line cap=butt,
        line width=0.800pt,
        draw,
        circle,
        minimum size=2pt,
      },
      font=\tiny
    ]
    \coordinate (p0) at (74,47);
    \coordinate (p1) at (42,27);
    \coordinate (p2) at (4,47);
    \coordinate (p3) at (38,42);
    \coordinate (p4) at (41,68);
    \coordinate (p5) at (38,52);

    \path[edge,red] (p0) -- (p1) -- (p2) -- (p5) -- cycle;
    \path[edge,blue] (p0) -- (p4) -- (p2) -- (p3) -- cycle;

    \node[vertex,label={[label distance=1mm]above:v}] at (p0) {};
    \node[vertex] at (p1) {};
    \node[vertex,label={[label distance=1mm]above:u}] at (p2) {};
    \node[vertex] at (p3) {};
    \node[vertex] at (p4) {};
    \node[vertex] at (p5) {};
\end{scope}

\end{tikzpicture}

     \end{adjustbox}
    \caption{$4$ paths all participate in separating $4$-cycles.}
    \label{fig:nontriv}
  \end{minipage}
  \hspace{.08\textwidth}
  \begin{minipage}[b]{.5\textwidth}
    \begin{adjustbox}{max width=\textwidth}
\begin{tikzpicture}[y=0.80pt, x=0.80pt, yscale=-1.000000, xscale=1.000000, inner sep=0pt, outer sep=0pt,
    base/.style={line join=miter,line cap=butt,line width=0.800pt,even odd rule},
    M edge/.style={base,draw=Purple,thick},
    K edge/.style={base,draw=Black},
    Y edge/.style={base,draw=Olive,thick},
    vertex/.style={draw,regular polygon,minimum size=5mm},
    K vertex/.style={
      base,
      draw=Black,
      fill=Black,
      vertex,
      circle,
      minimum size=3mm,
    },
    M vertex/.style={
      base,
      draw=Purple,
      fill=Magenta,
      vertex,
      regular polygon sides=3,
      minimum size=6mm,
    },
    Y vertex/.style={
      base,
      draw=Olive,
      fill=Yellow,
      vertex,
      regular polygon sides=4,
    },
  ]
  \begin{scope}[
    shift={(-27.64884,-259.1519)},
    ]
    \coordinate (p0) at (120,305);
    \coordinate (p1) at (676,671);
    \coordinate (p2) at (724,584);
    \coordinate (p3) at (610,525);
    \coordinate (p4) at (354,366);
    \coordinate (p5) at (58,422);
    \coordinate (p6) at (243,444);
    \coordinate (p7) at (614,292);
    \coordinate (p8) at (390,505);
    \coordinate (p9) at (78,576);
    \coordinate (p10) at (576,607);
    \coordinate (p11) at (314,586);
    \coordinate (p12) at (451,665);
    \coordinate (p13) at (218,644);
    \coordinate (p14) at (194,501);
    \coordinate (p15) at (568,385);
    \coordinate (p16) at (549,469);
    \coordinate (p17) at (472,504);
    \coordinate (p18) at (558,545);
    \coordinate (p19) at (141,500);
    \coordinate (p20) at (94,497);

    \path[M edge] (p5) -- (p6) -- (p0);
    \path[M edge] (p4) -- (p6) -- (p7);
    \path[M edge] (p7) -- (p8) -- (p5);
    \path[M edge] (p8) -- (p9);
    \path[M edge] (p8) -- (p10);
    \path[K edge] (p4) -- (p0) -- (p7);
    \path[Y edge] (p11) -- (p12);
    \path[Y edge] (p13) -- (p11);
    \path[Y edge] (p9) -- (p11);
    \path[Y edge] (p14) -- (p8);
    \path[Y edge] (p9) -- (p14);
    \path[Y edge] (p5) -- (p14);
    \path[Y edge] (p11) -- (p8);
    \path[K edge] (p10) -- (p1) -- (p2);
    \path[K edge] (p13) -- (p12) -- (p10) -- (p2) -- (p3) -- (p10);
    \path[K edge] (p7) -- (p3);
    \path[Y edge] (p15) -- (p7);
    \path[Y edge] (p15) -- (p16) -- (p17) -- (p8);
    \path[Y edge] (p17) -- (p18) -- (p7);
    \path[Y edge] (p8) -- (p15);
    \path[Y edge] (p18) -- (p16);
    \path[Y edge] (p10) -- (p18);
    \path[Y edge] (p8) -- (p18);
    \path[Y edge] (p5) -- (p19) -- (p9);
    \path[Y edge] (p5) -- (p20) -- (p9);
    \path[M edge] (p12) -- (p8);
    \path[K edge] (p0) -- (p5) -- (p9) -- (p13);

    \node[K vertex] at (p0) {};
    \node[K vertex] at (p1) {};
    \node[K vertex] at (p2) {};
    \node[K vertex] at (p3) {};
    \node[K vertex] at (p4) {};
    \node[K vertex] at (p5) {};
    \node[M vertex] at (p6) {};
    \node[K vertex] at (p7) {};
    \node[M vertex] at (p8) {};
    \node[K vertex] at (p9) {};
    \node[K vertex] at (p10) {};
    \node[Y vertex] at (p11) {};
    \node[K vertex] at (p12) {};
    \node[K vertex] at (p13) {};
    \node[Y vertex] at (p14) {};
    \node[Y vertex] at (p15) {};
    \node[Y vertex] at (p16) {};
    \node[Y vertex] at (p17) {};
    \node[Y vertex] at (p18) {};
    \node[Y vertex] at (p19) {};
    \node[Y vertex] at (p20) {};

  \end{scope}

\end{tikzpicture}

     \end{adjustbox}
    \caption{The sets $M$ (magenta), $Y$ (yellow), and $K$ (black) from Lemma~\ref{lem:QRS-partition}.}
    \label{fig:MYK}
  \end{minipage}
\end{figure}
\begin{lemma}\label{lem:QRS-partition}
	Given a planar graph $G=(V,E)$ and a vertex set $K\subseteq V$, let $M$ denote the vertices $m\in V\setminus K$ that have $d_K(m)\geq4$ (see Definition~\ref{def:quasisimple}), and let $Y$ denote $V\setminus(M\cup K)$. Then $Y$, $M$ and $K$ form a partition of $V$, and $\abs{M}\leq\abs{K}-2$.
\end{lemma}
\begin{proof}
	The partition property follows trivially from the definition.  Consider the maximal quasi-simple bipartite subgraph $H$ of $G$ with bipartition $(M, K)$.  By definition, each $m\in M$ has at least $4$ neighbors in the subgraph of $G$ quasi-induced by $K\cup\set{m}$.  Thus for each $m\in M$, $d_H(m)\geq 4$, and so $4\abs{M}\leq E(H)$.  Since $H$ is bipartite and quasi-simple, by Euler's formula we have $E(H)\leq2(\abs{M}+\abs{K})-4$.  Combining the two we get $4\abs{M}\leq E(H)\leq2(\abs{M}+\abs{K})-4$, which implies $\abs{M}\leq\abs{K}-2$.
\end{proof}

The aim of the data structure is to notice when new common neighbors of pairs of vertices in $K$ appear. The idea is that since $K$ has size $O(\sqrt{\abs{V}})$, there are only $O(|V|)$ pairs of vertices of $K$.
For each such pair, we maintain a doubly-linked list of all length-$2$ paths between them, and for each edge we maintain a doubly-linked list of all such paths it participates in.

When an edge $uw$ is inserted, new length-$2$ paths can only appear if $u\in K$ and/or $w\in K$. In this case the number of candidate paths to check is bounded by $d_K(u)+d_K(w)$, and this can be done in constant time per path.

When an edge $uw$ is contracted, new length-$2$ paths between vertices of $K$ may appear in other ways.
For example, we have new paths between neighbors of $u$ contained in $K$ and neighbors of $w$ in $K$.
Other cases are possible if $u$ or $w$ belongs to $K$.

We now define a potential function that decreases by at least the number of candidate paths after each contraction. We can decide if each candidate path is an actual length-$2$ path between vertices in $K$ in constant time.

It is defined in stages:
\begin{align*}
  \Phi_q (X)&= \sum_{v\in X} d_V(v)
  \\
   \Phi_v (X) &= 4\abs{X} - \frac{1}{2}\sum_{v\in X} d_X(v) = 4\abs{V(G_X)}-\abs{E(G_X)}
  \\
  \Phi_s(X) &= 63 (\Phi _v (X))^2 - \sum_{v\in X} (d_X(v))^2
  \\
  \Phi &=  6\Phi_v (V) + 3\Phi_q(Y\cup M) + \Phi_s(M\cup K)
\end{align*}%

\begin{restatable}{lemma}{potential}
  \label{lem:potential}
  The potential	 $\Phi$ is initially $O(n)$ and remains nonnegative.
  The embedding-respecting insertion or contraction of an edge $uv$ decreases $\Phi$ by at least 
  the number of candidate paths.
\end{restatable}

Before we proceed with the proof, we give a simple observation concerning quasi-simple graphs, which follows from Euler's formula and the fact that quasi-simple graphs of at least $3$ vertices have faces of degree at least $3$.
\begin{observation}\label{obs:factor3}
	In a quasi-simple planar graph with $n\ge 3$ vertices, the number of edges is at most $3n-6$.
\end{observation}

\begin{proof}[Proof of Lemma~\ref{lem:potential}]
To see the first statement, note that 
\begin{align*}
1\leq\abs{K}\leq\abs{M\cup K}\leq\Phi_v(M\cup K)\leq 4\abs{M\cup K}\leq4(2\abs{K}-2)=8(\abs{K}-1)
\end{align*}
where the inequality $\abs{M\cup K} \leq \Phi_v (M\cup K)$, stems from $\Phi_v(M\cup K)-\abs{M\cup K}=3\abs{V(G_{M\cup K})}-\abs{E(G_{M\cup K})}\geq0$, which is true by Observation~\ref{obs:factor3}. The last inequality is because there cannot be more than $|K|-2$ vertices in $M$ by Lemma~\ref{lem:QRS-partition}.
By a similar argument, we see that $\Phi_v(V)\geq 0$, and, being a sum of nonnegative terms, so is $\Phi_q(Y\cup M)\geq0$.

We may thus realize that $\Phi$ is always positive:
\begin{align*}
\Phi\geq \Phi_s(M\cup K) &\geq
63(\Phi_v(M\cup K))^2-\sum_{v\in M\cup K}(d_{M\cup K}(v))^2
\\
&\geq
63\abs{M\cup K}^2-\left(\sum_{v\in M\cup K}d_{M\cup K}(v)\right)^2
\\
&\geq
63\abs{M\cup K}^2-(6\abs{M\cup K}-12)^2
\\
&=
27\abs{M\cup K}^2+144(\abs{M\cup K}-1)
\\
&\geq 0
\end{align*}

Furthermore, note that $\Phi$ is initiated at $O(n)$. $\Phi_q (Y\cup M) \leq E$ which is $O(n)$ as the graph is planar. $\Phi_v(M\cup K) \leq 8 |K| = O(\sqrt n)$, and thus $(\Phi_v(M\cup K))^2 = O(n)$.

Before continuing with the second statement in the theorem, we consider how the different terms of $\Phi$ behave during changes.

Define $\Delta\Phi$, $\Delta\Phi_q$, $\Delta\Phi_v$ and $\Delta\Phi_s$ as the increases of the respective
values $\Phi$, $\Phi_q$, $\Phi_v$, $\Phi_s$ resulting from a change.

First we observe, that $\Phi_q (Y\cup M)$ has the following properties when contractions occur:
\begin{enumerate}
		\item $\Delta \Phi_q (Y\cup M) \leq -2$ when a pair of vertices in $Y\cup M$ 
		are contracted.\label{itm:qqmerge}
		\item $\Delta \Phi_q (Y\cup M) \leq -d_V(v)$ when a vertex $v\in Y\cup M$ is
		contracted with a vertex in $K$. \label{itm:qsmerge}
		\item $\Delta \Phi_q (Y\cup M)\leq 0$ when a pair of vertices in $K$ are
		contracted.\label{itm:ssmerge}
\end{enumerate}

Furthermore, we observe that for any $X\subseteq V$, $\Phi_v(X)$ has the following properties when changes occur:
\begin{enumerate}
		\item $\Delta \Phi_v(X) \leq 0$ when a vertex of degree $\geq 4$ is added to $X$. \label{itm:addr}
		\item $-4 \leq \Delta \Phi_v (X) \leq -1 $ when a vertex of degree $\leq 3$ is deleted from $X$.\label{itm:deleter}
		\item $\Delta \Phi_v(X) = -1 $ when an edge is added to $G_{X}$.\label{itm:addedge}
		\item $-3 \leq \Delta \Phi_v (X) \leq -1$ when contracting any edge and reducing to a quasi-simple graph.
		($\Delta\Phi_v (X) =-(3-\eta)$, where $0\leq \eta\leq 2$ is the number of additional edges deleted).\label{itm:contractandreduce}
\end{enumerate}
All the statements have similar proofs, so take for instance statement 4. Here, we decrease the first term by $4$ but increase the second by $\eta+1$, and thus, the resulting change is between $-3$ and $-1$. %

When an edge $uv$ is inserted, it can only create new length-$2$ paths between vertices of $K$ if at least one of its ends is in $K$. Suppose without loss of generality that $u\in K$. Then we have the following cases for where $v$ is before the insertion:
\begin{description}
\item[$v\in Y$:] Then $v$ had at most $3$ neighbors in $K$ before $uv$ was added, and thus at most $3$ candidate paths need to be checked. In this case $\Phi_v(V)$ drops by one, $\Phi_q(Y\cup M)$ increases by one, and $\Phi_s(M\cup K)$ is unchanged. Thus $\Phi$ drops by $3$.
\item[$v\in M$:] In this case there are $d_{K}(v)$ new candidate paths, $\Phi_v(V)$ drops by one and $\Phi_q(Y\cup M)$ increases by one just like before.  However, now $\Phi_v(M\cup K)$ drops by one, so $\Phi_v(M\cup K)^2$ drops by $2\Phi_v(M\cup K)-1$ and so $\Phi_s(M\cup K)$ drops by $63(2\Phi_v(M\cup K)-1)-(2d_{M\cup K}(u)-1)-(2d_{M\cup K}(v)-1)\geq 63(2\abs{M\cup K}-1)-(2(6\abs{M\cup K}-12)-1)-(2(6\abs{M\cup K}-12)-1)\geq 102\abs{M\cup K}+225$, which is much larger than $d_K(v)$.
\item[$v\in K$] In this case there are $d_K(u)+d_K(v)$ new candidate paths, $\Phi_v(V)$ drops by one and $\Phi_q(Y\cup M)$ is unchanged.  However, as in the previous case $\Phi_v(M\cup K)$ drops by one so $\Phi_s(M\cup K)$ drops by at least $102\abs{M\cup K}+225$, which is much larger than $d_K(u)+d_K(v)$.
\end{description}

Finally we consider the case where an edge $uv$ is contracted.
	To check the lemma,
	one simply has to check the different combinations
	of which partition the two elements and the result belong to:
	
	\begin{description}
		\item[$(Y,Y)$ merge to $Y$:] $\Phi_q(Y\cup M)$ drops by at least $2$, and
		the other terms in the potential are unchanged, so $\Phi$ drops
		by at least $6$.  Since the result $v$ is in $Y$, there are at
		most $3$ paths of length $2$ in $G_{K\cup\set{v}}$ with $v$ as a middle vertex, and at most
		$2$ of them are new.
		
		\item[$(Y,Y)$ merge to $M$:] The product of the degrees, and therefore the number of candidate paths is at most
		$9$. $\Delta\Phi_v(M\cup K)\leq0$ and $\Delta\Phi_q(Y\cup M)\leq-2$, and
		the term $\sum_{v\in M\cup K}(d_{M\cup K}(v))^2$ drops by
		the sum of degrees squared.
		Thus, $\Delta\Phi\leq -2+0-8<-9$,
		and we are done.
		
		\item[$(M,Y)$-merge] Suppose $u\in M$ and $v\in Y$.  Let $w$
		be the node that $u,v$ are contracted to, it will be an $M$-node.
		
                We call an edge $(v,k)$ important if it participates in a new length $2$ path between different vertices of $K$.
		Each important edge incident to $v\in Y$ will become part of the
		graph quasi-induced by $M\cup K$, and therefore cause a drop in
		$\Phi_v(M\cup K)$. This drop is enough to pay for all new paths containing
		that edge.
		
		\item[$(K,Y)$-merge] Suppose $u\in K$ and $v\in Y$.  Let $w$ be the
		node that $u,v$ are contracted to, it will be an $K$-node.
		
		There are two types of new length-two $(K,K)$ paths that arise:
		Paths having $w$ as the middle vertex, and paths having $w$ as an
		end vertex.
		
		The paths with $w$ as a middle vertex are accounted for just like
		the previous case.
		
		Each new path having $w$ as an end vertex must have a neighbor of
		$v$ as middle vertex.  There are (less than) $d_V(v)$ of these
		neighbors. Since $\Delta\Phi_q(Y\cup M)\leq-d_V(v)$ we can afford to look
		at each of them, and pay for at most $2$ new paths for each.
		
		Now consider a neighbor $m$ of $v$ that is middle vertex of some
		new path.  If $m\in Y$ (after the contraction), then there is at
		most $2$ new paths involving $m$, and the drop in $\Phi_q(Y\cup M)$ pays
		for them.  If $m\in M\cup K$, then $d_{M\cup K}(m)$ has increased, and
		$\Phi_s(M\cup K)$ drops appropriately.
		
		\item[$(M,M)$-merge] Suppose $u,v\in M$ are merged to the new vertex
		$w$. Note that $w\in M$. Let $X=\set{x_1,\ldots,x_k}$ be the set
		of common neighbors of $u,v$ in $G_{M\cup K}$ that lose an edge
		when quasi-simplifying after the contraction, and note that $0\leq
		\eta\leq 2$. Let $\Phi_v=\Phi_v(M\cup K)$, then
		\begin{align*}
		\Delta(\Phi_v^2)
		=
		(\Phi_v+\Delta\Phi_v)^2-\Phi_v^2
		=
		(\Phi_v-(3-\eta))^2-\Phi_v^2
		=
		(3-\eta)^2-2(3-\eta)\Phi_v
		\end{align*}
		Let $a=d_{M\cup K}(u)$, $b=d_{M\cup K}(v)$, and for
		$i\in\set{1,\ldots,\eta}$ let $c_i=d_{M\cup K}(x_i)$.  Then
		\begin{align*}
                  a,b,c_i
                  \leq a+b+\sum_{i=1}^\eta c_i
                  \leq \sum_{y\in M\cup K}d_{M\cup K}(y)
                  \leq 6\abs{M\cup K}-12
                  \leq 6\Phi_v-12
		\end{align*}
                And finally
		\begin{align*}
		\Delta\Phi_s(M\cup K)
		&\leq
		63\Delta(\Phi_v^2)
		\\ &\qquad
		-
		\left(
		\left((a+b-(\eta+2))^2+\sum_{i=1}^{\eta}(c_i-1)^2\right)
		-
		\left(a^2+b^2+\sum_{i=1}^{\eta}c_i^2\right)
		\right)
		\\
		&=
		63((3-\eta)^2-2(3-\eta)\Phi_v)
		\\ &\qquad
		-
		\left(2ab + (\eta+2)^2 - 2(\eta+2)a - 2(\eta+2)b - \sum_{i=1}^\eta(2c_i-1)\right)
		\\
		&=
		63((3-\eta)^2-2(3-\eta)\Phi_v)
		\\ &\qquad
		-
		2ab
		-
		((\eta+2)^2+\eta)
		+
		\left(2(\eta+2)a + 2(\eta+2)b + \sum_{i=1}^{\eta} 2c_i\right)
		\\
		&\leq
		63((3-\eta)^2-2(3-\eta)\Phi_v)
		\\ &\qquad
		-
		2ab
		-
		((\eta+2)^2+\eta)
		+
		\left(2(\eta+2) + 2(\eta+2) + \sum_{i=1}^{\eta}2\right)(6\Phi_v-12)
		\\
		&=
		63((3-\eta)^2-2(3-\eta)\Phi_v)
		\\ &\qquad
		-
		2ab
		-
		((\eta+2)^2+\eta)
		+
		(6\eta +8)(6\Phi_v-12)
		\\
		&=
		- 2ab +
		\begin{cases}
		467-330\Phi_v &\text{if }\eta=0 \\
		74-168\Phi_v  &\text{if }\eta=1 \\
		-195-6\Phi_v  &\text{if }\eta=2
		\end{cases}
		\\
		&\leq
		- 2ab - 6\Phi_v \qquad\text{for }\Phi_v\geq2
		\stepcounter{equation}\tag{\theequation}\label{eq:PhiDrop-RR}
		\end{align*}
		In particular, $\Delta\Phi \leq -ab = -d_{M\cup K}(u)\cdot d_{M\cup K}(v) \leq - d_S(u)\cdot d_S(v)$, as desired.
		\item[$(K,M)$-merge:] Suppose $u\in K$ and $v\in M$.
		Let $w$ be the node that $u,v$ are contracted to, it will be an
		$K$-node.
		
		There are two types of new length-two $(K,K)$ paths that arise:
		Paths having $w$ as the middle vertex, and paths having $w$ as an
		end vertex.
		
		The paths with $w$ as a middle vertex are accounted for just like
		the previous case.
		
		Each new path having $w$ as an end vertex must have a neighbor of
		$v$ as middle vertex.  There are (less than) $d_V(v)$ of these
		neighbors. Since $\Delta\Phi_q(Y\cup M)\leq-d_V(v)$ we can afford to look
		at each of them, and pay for at most $2$ new paths for each.
		
		Now consider a neighbor $m$ of $v$ that is middle vertex of some
		new path.  If $m\in Y$ (after the contraction), then there is at
		most $2$ new paths involving $m$, and the drop in $\Phi_q$ pays
		for them.
		
		Let $M$ be the set of neighbors of $v$ in $M\cup K$ that is middle
		of some new path. The number of such paths is at most
		$\abs{E(G_{M\cup K})}\leq3\abs{M\cup K}-6\leq 3\Phi_v(M\cup K)-6$,
		since each must contain a unique edge from $G_{M\cup K}$.
		And the $-6\Phi_v (M\cup K)$ pays for them.
		
		\item[$(K,K)$-merge:] Suppose $u,v\in K$.  Let $w$ be the node that
		$u,v$ are contracted to, it will be an $K$-node.
		
		There are two types of new length-two $(K,K)$ paths that arise:
		Paths having $w$ as the middle vertex, and paths having $w$ as an
		end vertex.
		
		The paths with $w$ as a middle vertex are accounted for just like
		the previous two cases.
		
		Each new path having $w$ as an end vertex was already an $(K,K)$
		path with either $u$ or $v$ before the merge.  For each of $u,v$
		there is at most $\abs{K}-1$ such pairs that (may) need to be
		updated, so the total cost of updating these is less than
		$2(\abs{K}-1)\leq 2\Phi_v$.  The $-6\Phi_v$ can pay for
		these updates.\qedhere
	\end{description}
\end{proof}

As a result, Lemma~\ref{lem:ineachnode} follows, and we are finally ready to prove Theorem~\ref{thm:4cycledetection}.
\begin{proof}[Proof of Theorem~\ref{thm:4cycledetection}]
  Given a planar graph $G$ with bounded face-degree, we build a face-preserving separator tree as in Lemma~\ref{lem:septree} in $O(n\log n)$ time.
  For each internal vertex of the tree, we may detect new $4$-cycles crossing the separator due to Lemma~\ref{lem:ineachnode}. 
  The leaves have size at most $n_0 = O(1)$, and we can detect $4$-cycles in the leaves in $O(1)$ time.
	
	We can distinguish between face $4$-cycles and separating $4$-cycles. Namely, we can choose to report only when an edge first lies on any $4$-cycle, or when it first lies on a separating one, as described in Lemma~\ref{lem:ineachnode}.

        An edge insertion in the graph $H_x$ needs to be duplicated in each child of $x$ that contains both vertices. However, the drop in the $\Phi$ potentials for each node is large enough to pay for each cascading insertion.
	Whenever a contraction in the graph $H_x$ for a node $x$ of the separator tree introduces a new edge between two separator vertices, that edge may need to be added to the subtree containing the other side of the separation, but again that is paid for. In general, if we update graphs closer to the root first, the changes only propagate down and every change is paid for by a corresponding drop in the potential.
\end{proof}

\section{Omitted proofs from Section~\ref{sec:SPQR}}\label{sec:SPQR-proofs}

\fvCycleEdgeclasses*
\begin{proof}
	Let $C$ be the $4$-cycle. Consider a path $L$ in $G$ containing edges $e_1$ and $e_2$.  Consider the set of faces $F$ in $\fv{G}$ that are incident to a vertex on $L$.  $L$ does not cross $v_1,v_2$ if and only if $F$ is completely contained on one side of $C$, which happens if and only if $e_1$ and $e_2$ are not separated by $C$. An identical argument can be made about $f_1,f_2$ in $\dual{G}$
\end{proof}

\sepone*
\begin{proof}
    If $C$ is a separating cycle, there is at least two faces on either side. 
By Lemma~\ref{lem:fvCycleEdgeclasses} there are thus at least two different separation classes with respect to $v_1,v_2$ (or $f_1,f_2$). If there exactly $2$ separation classes, each consists of at least two edges. If there are exactly $3$ classes, at least one of them consists of at least two edges.
\end{proof}

\classescircular*
\begin{proof}
    The proof is by contradiction.
    Assume that the circular order of some $4$ edges incident to $u$ is $e_1, e_2, e_3, e_4$.
    Moreover, assume that only $e_1$ and $e_3$ belong to the same separation class.
    From Definition~\ref{def:separationpair} there is a path that begins with $e_1$ and ends with $e_3$ that does not contain $u$ or $w$ as its internal point.
    Thus, this path is a cycle $C$ that does not go through $w$.
    Hence, every path from either $e_2$ or $e_4$, that ends in $w$ and does not contain $u$ as an internal point, has to go through $C$.
    This contradicts the fact that $e_2$ and $e_4$ are in different separation classes than $e_1$ and $e_3$.
 Clearly,  the same argument applies to $w$.
\end{proof}

\begin{restatable}{lemma}{delimitF}\label{lem:delimitF}
    Let $G$ be a loopless biconnected plane graph.
    Let $F$ be a subset of edges of $G$, such that $F$ is a separation class for some pair $u, w$ of vertices of $G$.
    Then there exists a $4$-cycle (possibly non-separating) in $\fv{G}$ that separates the set of faces that correspond to $F$ from all other faces.
\end{restatable}

\begin{proof}%
    Throughout the proof by separation class we mean one of separation classes defined by $u$ and $w$.
    Clearly, each separation class has to have an edge incident to $u$ or $w$ (otherwise, since the graph is connected, it would not be maximal).
    In fact, since $G$ is biconnected, each separation class has edges incident both to $u$ and $w$.
    If a separation class had edges incident only to one of the two vertices, this vertex would be an articulation point.
     
    Denote the edges incident to $u$ in circular order by $e_1, \ldots, e_k$.
    For convenience, let $e_{k+1} := e_1$ and $e_{0} := e_k$.
	Moreover, assume that $e_i \in F$ iff $a \leq i \leq b$, where $1 \leq a \leq b \leq k$.
    Note that by Lemma~\ref{lem:classescircular}, $a$ and $b$ are well-defined (for some way of breaking the circular ordering into a sequence $e_1, \ldots, e_k$).

    Let $f_1$ be the face that comes in the circular order between $e_{a-1}$ and $e_a$ and $f_2$ be the face that comes between $e_b$ and $e_{b+1}$.
    Note that $f_1 \neq f_2$.
    
    We now show that there is a $4$-cycle in $\fv{G}$ that contains $u$, $w$, $f_1$ and $f_2$.
    To that end, we prove that both  $u$ and $w$ lie on $f_1$.
    Indeed, the cycle bounding $f_1$ is simple and contains edges from two separation classes.
    Thus, it has to contain both $u$ and $w$.
    Similarly, both $u$ and $w$ lie on $f_2$.
    
    This implies that $\fv{G}$ contains a $4$-cycle $C_F$ going through $u$, $w$, $f_1$ and $f_2$, and by construction, $C_F$ separates the faces corresponding to $F$ from all other edges.
\end{proof}

\begin{restatable}{lemma}{seppair}\label{lem:seppair}
  Let $G$ be loopless biconnected graph.
  If $v_1,v_2$ is a separation pair in $G$, then there exists a separation pair $f_1,f_2$ in $\dual{G}$ such that $(v_1,f_1,v_2,f_2)$ is a separating cycle in $\fv{G}$.
\end{restatable}
\begin{proof}
	If every separation class with respect to $v_1,v_2$ consists of a single edge, then $G$ consists of two vertices connected by multiple edges and the lemma is trivial. It suffices to use the fact that since $v_1,v_2$ is a separation pair, there are at least $4$ edges in $G$.
	Otherwise there is a separation class $F$ with at least two edges, such there are at least two edges in $E(G)\setminus F$.  Now apply Lemma~\ref{lem:delimitF}, to get a delimiting cycle $C$ in $\fv{G}$ that separates faces corresponding to $F$ from all other faces.
Denote the vertices of $C$ by $v_1,f_1,v_2,f_2$.
	Since both $F$ and $E(G)\setminus F$ are nontrivial (both contain more than one edge), $C$ is a separating cycle in $\fv{G}$. It then follows from Lemma~\ref{lem:sep1} that $f_1,f_2$ form a separation pair in $\dual{G}$.
\end{proof}

\DeleteContractSeq*
\begin{proof}
For each neighbor $y$ of $x$ in the SPQR-tree $T$ we proceed as follows.
Let $a,b$ be the separation pair corresponding to the edge in $T$ between $x$ and $y$.
Consider all nodes reachable from $y$ in $T$ with a path that does not contain $x$.
Let $D$ be the set of non-virtual edges in all these nodes. While there is an edge $e$ in $D$ that is not a self-loop and not an edge between $a$ and $b$, contract it. Then if there are any self-loops delete them. When all edges in $D$ go between $a$ and $b$, delete edges until there is only one left.
\end{proof}

\spqrpath*

\begin{proof}
Let us first prove that $H$ is indeed a path.
Since $G-e$ is biconnected, there exist two internally vertex-disjoint paths between $u$ and $w$.
  No separation pair in $G - e$ can have both vertices on the same of these paths, since otherwise it would be a separation pair in $G$.
Moreover, observe that each separation pair defines at most two separation classes that consist of more than one edge (otherwise, it is also a separation pair in $G$).
Thus, we can split $G-e$ into two subgraphs by using an arbitrary separation pair in $G-e$.
By repeating the same reasoning on both subgraphs, we get that $H$ is a path.
Observe that $u$ and $w$ belong to the nodes at the opposite ends of $H$.

Note that since we know the edges belonging to $4$-cycles in $\fv{(G-e)}$, by using Lemma~\ref{lem:fvCycleEdgeclasses} we also know all separation pairs in $G-e$.
We now describe how to compute all components of $H$ (i.e. the skeleton graphs stored in the nodes) except the largest one.
Consider an algorithm that starts from one end of $H$ and discovers the components one by one, each in linear time.
Observe that if each component of $H$ has size at most $|E(G)|/2$, we can afford to detect all components of $H$, without affecting the total running time.
However, to prepare for the opposite case, we need to do the search in parallel, starting from both ends of $H$.
Let $y$ be the largest component of $H$.
Observe that only one search can start exploring edges of $y$.
As soon as the other search reaches $y$, we have discovered all separation pairs (that is why we need to know all separation pairs upfront), and both searches can stop.
Thus, one of the searches only uses time that is at most the total size of all non-largest components of $H$.
Since the other search runs in parallel, it runs in the same asymptotic time.

To complete the proof it remains to describe how the search procedure works.
Recall that by Lemma~\ref{lem:classescircular}, for each separation pair $u, w$ of $G - e$, the edges belonging to each separation class come in consecutive order around $u$ and $w$.
Observe that the edges of the $4$-cycles of $\fv{(G-e)}$ correspond to the corners of $G-e$ that lie between edges belonging to distinct separation classes.
Thus, once we mark these corners in $G-e$, we can run a DFS-search that, once started from an edge belonging to a skeleton graph of an $S$- or $R$-node, explores all edges of this graph (and only those).

Observe moreover, that from our earlier analysis it follows that the endpoints $u$ and $w$ of $e$ do not belong to any separation pair.
This implies that both $u$ and $w$ are contained in $S$- or $R$- nodes.
Let us focus on the search starting from $u$.
Note that it discovers the entire component containing $u$ and the separation pairs that separates it from the rest of the SPQR-tree.
If the skeleton graph of this component is a path connecting the vertices of the separation pair, we have found an $S$-node.
Otherwise, we have found an $R$-node.

Now assume we have found some prefix of the SPQR-path that ends at a separation pair $a, b$.
If there is an edge $ab$ (this edge comes next in the circular ordering after the edges we have visited, so it is easy to find), the next node on the SPQR-path is a $P$-node.
After we have processed all edges between $a$ and $b$, we insert a virtual edge between $a$ and $b$ and continue the search starting from this edge in a similar way to the search that has discovered the first node on the SPQR-path.
Clearly, the algorithm runs in linear time.
\end{proof}

We now show that algorithms~\ref{alg:removep}, \ref{alg:remover}, \ref{alg:removes} are correct.
In each proof, the goal is to show that after the procedure the tree $T$ satisfies Definition~\ref{def:SPQR}.

\begin{lemma}\label{lem:removep}
Algorithm~\ref{alg:removep} is correct.
\end{lemma}

\begin{proof}
If after the edge deletion, $\Gamma(x)$ sill has at least $3$ edges, then clearly $T$ is a valid SPQR-tree.
Otherwise, $\Gamma(x)$ has exactly two edges and we consider three cases.
Recall that the number of virtual edges in a node is equal to the number of the node's neighbors in the SPQR-tree.
If $\Gamma(x)$ has no virtual edges, then $x$ is the only node of $T$, and thus this biconnected component now only has $2$ edges, so we should delete the entire SPQR-tree.
If $\Gamma(x)$ has exactly one virtual edge, $x$ has exactly one neighbor.
In this case, simply $x$ represents one edge of the graph, so it has to be merged with its only neighbor and the virtual edge in the neighbor becomes non-virtual.
If $\Gamma(x)$ has two virtual edges we remove $x$ and the neighbors of $x$ become neighbors.
Note that the neighbors of $x$ cannot be $P$-nodes.
Thus, unless $x$ has two neighboring $S$-nodes, we obtain a valid SPQR-tree.
In the remaining case, it is easy to see that the two $S$-nodes can be merged into one $S$-node.
\end{proof}

\begin{lemma}\label{lem:remover}
Algorithm~\ref{alg:remover} is correct.
\end{lemma}

\begin{proof}
If after removing the edge, $\Gamma(x)$ is triconnected, clearly the tree is a valid SPQR-tree.
Otherwise, by Lemma~\ref{lem:spqrpath}, $\Gamma(x)$ is represented by a SPQR-path.
It is easy to see that after replacing $x$ by the SPQR-path $X'$, we obtain a valid SPQR-tree, unless there are two neighboring $S$-nodes or $P$-nodes.
Since the SPQR-path is a SPQR-tree, each such pair contains exactly one node $z$ from the SPQR-path. If it is a $P$ node, it can not be the end of the path, and so has at least $2$ virtual edges to its neighbors on the path, and at most one more virtual edge to a neighbor $z'$ outside the path. If it is an $S$-node it may have up to $2$ virtual edges to a neighbor $z'$ outside the path.
Clearly, if $z$ and $z'$ have the same type, they can be merged into one node, and this yields a valid SPQR-tree.
\end{proof}

\begin{lemma}\label{lem:removes}
Algorithm~\ref{alg:removes} is correct.
\end{lemma}

\begin{proof}
Observe that after removing an edge $e = uw$, each vertex of $\Gamma(x)$ distinct from $u$ and $w$ is an articulation point.
Thus, each neighbor of $x$ now belongs to a different biconnected component.
Thus, we update $T$ by deleting $x$, which breaks $T$ into a piece for each neighbor $y$. For each piece we create a new BC-node $z$.

Each non-virtual edge of $\Gamma(x)$ becomes a biconnected component by itself, so we could simply ignore these edges from now on.
For each virtual edge of $\Gamma(x)$, we delete the corresponding edge in the neighbor of $x$ using an appropriate function.
From Lemmas~\ref{lem:removep} and~\ref{lem:remover} it follows that the SPQR-trees are updated correctly.
\end{proof}

\bibliographystyle{plain}
\bibliography{references}

\begin{thebibliography}{10}

\bibitem{AD16}
Amir Abboud and S{\o}ren Dahlgaard.
\newblock Popular conjectures as a barrier for dynamic planar graph algorithms.
\newblock In {\em {IEEE} 57th Annual Symposium on Foundations of Computer
  Science, {FOCS} 2016, 9-11 October 2016, Hyatt Regency, New Brunswick, New
  Jersey, {USA}}, pages 477--486, 2016.

\bibitem{AW14}
Amir Abboud and Virginia~Vassilevska Williams.
\newblock Popular conjectures imply strong lower bounds for dynamic problems.
\newblock In {\em 55th {IEEE} Annual Symposium on Foundations of Computer
  Science, {FOCS} 2014, Philadelphia, PA, USA, October 18-21, 2014}, pages
  434--443, 2014.

\bibitem{Abraham2012}
Ittai Abraham, Shiri Chechik, and Cyril Gavoille.
\newblock Fully dynamic approximate distance oracles for planar graphs via
  forbidden-set distance labels.
\newblock In {\em Proceedings of the 44th Symposium on Theory of Computing
  Conference, {STOC} 2012, New York, NY, USA, May 19 - 22, 2012}, pages
  1199--1218, 2012.

\bibitem{Angelini2013}
Patrizio Angelini, Thomas Bl{\"{a}}sius, and Ignaz Rutter.
\newblock Testing mutual duality of planar graphs.
\newblock {\em Int. J. Comput. Geometry Appl.}, 24(4):325--346, 2014.

\bibitem{ARCHDEACON199237}
Dan Archdeacon and R~Bruce Richter.
\newblock The construction and classification of self-dual spherical polyhedra.
\newblock {\em Journal of Combinatorial Theory, Series B}, 54(1):37 -- 63,
  1992.

\bibitem{bgs2015}
Surender Baswana, Manoj Gupta, and Sandeep Sen.
\newblock Fully dynamic maximal matching in ${O}(\log n)$ update time.
\newblock {\em {SIAM} J. Comput.}, 44(1):88--113, 2015.

\bibitem{doi:10.1137/0406017}
Graham~R. Brightwell and Edward~R. Scheinerman.
\newblock Representations of planar graphs.
\newblock {\em SIAM Journal on Discrete Mathematics}, 6(2):214--229, 1993.

\bibitem{Brinkmann:2005:GSQ:2651845.2652314}
Gunnar Brinkmann, Sam Greenberg, Catherine Greenhill, Brendan~D. Mckay, Robin
  Thomas, and Paul Wollan.
\newblock Generation of simple quadrangulations of the sphere.
\newblock {\em Discrete Math.}, 305(1-3):33--54, December 2005.

\bibitem{CHILP16}
Shiri Chechik, Thomas~Dueholm Hansen, Giuseppe~F. Italiano, Jakub {{\L}{\c
  a}cki}, and Nikos Parotsidis.
\newblock Decremental single-source reachability and strongly connected
  components in $\widetilde{O}(m\sqrt{n}\,)$ total update time.
\newblock In {\em {IEEE} 57th Annual Symposium on Foundations of Computer
  Science, {FOCS} 2016, 9-11 October 2016, Hyatt Regency, New Brunswick, New
  Jersey, {USA}}, pages 315--324, 2016.

\bibitem{DI04}
Camil Demetrescu and Giuseppe~F. Italiano.
\newblock A new approach to dynamic all pairs shortest paths.
\newblock {\em J. {ACM}}, 51(6):968--992, 2004.

\bibitem{DI08}
Camil Demetrescu and Giuseppe~F. Italiano.
\newblock Mantaining dynamic matrices for fully dynamic transitive closure.
\newblock {\em Algorithmica}, 51(4):387--427, 2008.

\bibitem{Battista:96}
Giuseppe {Di Battista} and Roberto Tamassia.
\newblock On-line maintenance of triconnected components with {SPQR}-trees.
\newblock {\em Algorithmica}, 15(4):302--318, 1996.

\bibitem{BT96}
Giuseppe {Di Battista} and Roberto Tamassia.
\newblock On-line planarity testing.
\newblock {\em {SIAM} J. Comput.}, 25(5):956--997, 1996.

\bibitem{Diks2007}
Krzysztof Diks and Piotr Sankowski.
\newblock Dynamic plane transitive closure.
\newblock In {\em Algorithms - {ESA} 2007, 15th Annual European Symposium,
  Eilat, Israel, October 8-10, 2007, Proceedings}, pages 594--604, 2007.

\bibitem{EGIN97}
David Eppstein, Zvi Galil, Giuseppe~F. Italiano, and Amnon Nissenzweig.
\newblock Sparsification - a technique for speeding up dynamic graph
  algorithms.
\newblock {\em J. {ACM}}, 44(5):669--696, 1997.

\bibitem{Eppstein96}
David Eppstein, Zvi Galil, Giuseppe~F. Italiano, and Thomas~H. Spencer.
\newblock Separator based sparsification {I}: {P}lanarity testing and minimum
  spanning trees.
\newblock {\em J. Comput. Syst. Sci.}, 52(1):3--27, 1996.

\bibitem{Eppstein:1998}
David Eppstein, Zvi Galil, Giuseppe~F. Italiano, and Thomas~H. Spencer.
\newblock Separator-based sparsification {II:} {E}dge and vertex connectivity.
\newblock {\em {SIAM} J. Comput.}, 28(1):341--381, 1998.
\newblock Announced at STOC '93.

\bibitem{Eppstein92}
David Eppstein, Giuseppe~F. Italiano, Roberto Tamassia, Robert~Endre Tarjan,
  Jeffery Westbrook, and Moti Yung.
\newblock Maintenance of a minimum spanning forest in a dynamic plane graph.
\newblock {\em J. Algorithms}, 13(1):33--54, 1992.

\bibitem{FR06}
Jittat Fakcharoenphol and Satish Rao.
\newblock Planar graphs, negative weight edges, shortest paths, and near linear
  time.
\newblock {\em J. Comput. Syst. Sci.}, 72(5):868--889, 2006.

\bibitem{Giammarresi:96}
Dora Giammarresi and Giuseppe~F. Italiano.
\newblock Decremental 2- and 3-connectivity on planar graphs.
\newblock {\em Algorithmica}, 16(3):263--287, 1996.
\newblock Announced at SWAT 1992.

\bibitem{Gustedt98}
Jens Gustedt.
\newblock Efficient union-find for planar graphs and other sparse graph
  classes.
\newblock {\em Theor. Comput. Sci.}, 203(1):123--141, 1998.

\bibitem{Gutwenger2001}
Carsten Gutwenger and Petra Mutzel.
\newblock {\em A Linear Time Implementation of SPQR-Trees}, pages 77--90.
\newblock Springer Berlin Heidelberg, Berlin, Heidelberg, 2001.

\bibitem{Harary69}
Frank Harary.
\newblock {\em Graph Theory}.
\newblock Addison-Wesley Series in Mathematics. Addison Wesley, 1969.

\bibitem{HeTh97}
Monika~Rauch Henzinger and Mikkel Thorup.
\newblock Sampling to provide or to bound: With applications to fully dynamic
  graph algorithms.
\newblock {\em Random Struct. Algorithms}, 11(4):369--379, 1997.

\bibitem{HoLiTh01}
Jacob Holm, Kristian de~Lichtenberg, and Mikkel Thorup.
\newblock Poly-logarithmic deterministic fully-dynamic algorithms for
  connectivity, minimum spanning tree, 2-edge, and biconnectivity.
\newblock {\em J. {ACM}}, 48(4):723--760, 2001.

\bibitem{2017arXiv170610228H}
Jacob {Holm}, Giuseppe~F. {Italiano}, Adam {Karczmarz}, Jakub {{\L}{\c a}cki},
  Eva {Rotenberg}, and Piotr {Sankowski}.
\newblock {Contracting a Planar Graph Efficiently}.
\newblock {\em ArXiv e-prints}, June 2017.
\newblock \url{https://arxiv.org/abs/1706.10228v1} Accepted for ESA 2017.

\bibitem{Holm2017}
Jacob Holm and Eva Rotenberg.
\newblock Dynamic planar embeddings of dynamic graphs.
\newblock {\em Theory of Computing Systems}, Apr 2017.

\bibitem{Holm15}
Jacob Holm, Eva Rotenberg, and Christian Wulff{-}Nilsen.
\newblock Faster fully-dynamic minimum spanning forest.
\newblock In {\em Algorithms - {ESA} 2015 - 23rd Annual European Symposium,
  Patras, Greece, Sept. 14-16, 2015, Proceedings}, pages 742--753, 2015.

\bibitem{hopcroft1973dividing}
John~E. Hopcroft and Robert~Endre Tarjan.
\newblock Dividing a graph into triconnected components.
\newblock {\em SIAM Journal on Computing}, 2(3):135--158, 1973.

\bibitem{Hopcroft73}
John~E. Hopcroft and Robert~Endre Tarjan.
\newblock A {V} log {V} algorithm for isomorphism of triconnected planar
  graphs.
\newblock {\em J. Comput. Syst. Sci.}, 7(3):323--331, 1973.

\bibitem{Hopcroft74}
John~E. Hopcroft and J.~K. Wong.
\newblock Linear time algorithm for isomorphism of planar graphs (preliminary
  report).
\newblock In {\em Proceedings of the 6th Annual {ACM} Symposium on Theory of
  Computing, April 30 - May 2, 1974, Seattle, Washington, {USA}}, pages
  172--184, 1974.

\bibitem{IKLS17}
Giuseppe~F. Italiano, Adam Karczmarz, Jakub {{\L}{\c a}cki}, and Piotr
  Sankowski.
\newblock Decremental single-source reachability in planar digraphs.
\newblock In {\em Proceedings of the 49th Annual {ACM} {SIGACT} Symposium on
  Theory of Computing, {STOC} 2017, Montreal, QC, Canada, June 19-23, 2017},
  pages 1108--1121, 2017.

\bibitem{INSW11}
Giuseppe~F. Italiano, Yahav Nussbaum, Piotr Sankowski, and Christian
  Wulff{-}Nilsen.
\newblock Improved algorithms for min cut and max flow in undirected planar
  graphs.
\newblock In {\em Proceedings of the 43rd {ACM} Symposium on Theory of
  Computing, {STOC} 2011, San Jose, CA, USA, 6-8 June 2011}, pages 313--322,
  2011.

\bibitem{kant2001}
Goossen Kant.
\newblock Algorithms for drawing planar graphs.
\newblock 2001.

\bibitem{KMNS12}
Haim Kaplan, Shay Mozes, Yahav Nussbaum, and Micha Sharir.
\newblock Submatrix maximum queries in {M}onge matrices and {M}onge partial
  matrices, and their applications.
\newblock In {\em Proceedings of the Twenty-Third Annual {ACM-SIAM} Symposium
  on Discrete Algorithms, {SODA} 2012, Kyoto, Japan, January 17-19, 2012},
  pages 338--355, 2012.

\bibitem{KaKiMo13}
Bruce~M. Kapron, Valerie King, and Ben Mountjoy.
\newblock Dynamic graph connectivity in polylogarithmic worst case time.
\newblock In {\em Proceedings of the Twenty-Fourth Annual {ACM-SIAM} Symposium
  on Discrete Algorithms, {SODA} 2013, New Orleans, Louisiana, USA, January
  6-8, 2013}, pages 1131--1142, 2013.

\bibitem{KR16}
Casper Kejlberg{-}Rasmussen, Tsvi Kopelowitz, Seth Pettie, and Mikkel Thorup.
\newblock Faster worst case deterministic dynamic connectivity.
\newblock In {\em 24th Annual European Symposium on Algorithms, {ESA} 2016,
  August 22-24, 2016, Aarhus, Denmark}, pages 53:1--53:15, 2016.

\bibitem{K99}
Valerie King.
\newblock Fully dynamic algorithms for maintaining all-pairs shortest paths and
  transitive closure in digraphs.
\newblock In {\em 40th Annual Symposium on Foundations of Computer Science,
  {FOCS} '99, 17-18 October, 1999, New York, NY, {USA}}, pages 81--91, 1999.

\bibitem{K05}
Philip~N. Klein.
\newblock Multiple-source shortest paths in planar graphs.
\newblock In {\em Proceedings of the Sixteenth Annual {ACM-SIAM} Symposium on
  Discrete Algorithms, {SODA} 2005, Vancouver, BC, Canada, January 23-25,
  2005}, pages 146--155, 2005.

\bibitem{Klein:book}
Philip~N. Klein and Shay Mozes.
\newblock Optimization algorithms for planar graphs, 2017.

\bibitem{Klein:13}
Philip~N. Klein, Shay Mozes, and Christian Sommer.
\newblock Structured recursive separator decompositions for planar graphs in
  linear time.
\newblock In {\em Symposium on Theory of Computing Conference, STOC'13, Palo
  Alto, CA, USA, June 1-4, 2013}, pages 505--514, 2013.

\bibitem{scc-decomposition}
Jakub {{\L}{\c a}cki}.
\newblock Improved deterministic algorithms for decremental reachability and
  strongly connected components.
\newblock {\em {ACM} Trans. Algorithms}, 9(3):27:1--27:15, 2013.

\bibitem{steiner-tree}
Jakub {{\L}{\c a}cki}, Jakub O\'{c}wieja, Marcin Pilipczuk, Piotr Sankowski,
  and Anna Zych.
\newblock The power of dynamic distance oracles: Efficient dynamic algorithms
  for the {S}teiner tree.
\newblock In {\em Proceedings of the Forty-Seventh Annual {ACM} on Symposium on
  Theory of Computing, {STOC} 2015, Portland, OR, USA, June 14-17, 2015}, pages
  11--20, 2015.

\bibitem{Lacki2011}
Jakub {{\L}{\c a}cki} and Piotr Sankowski.
\newblock Min-cuts and shortest cycles in planar graphs in ${O}(n\log\log{n})$
  time.
\newblock In {\em Algorithms - {ESA} 2011 - 19th Annual European Symposium,
  Saarbr{\"{u}}cken, Germany, September 5-9, 2011. Proceedings}, pages
  155--166, 2011.

\bibitem{decremental-connectivity}
Jakub {{\L}{\c a}cki} and Piotr Sankowski.
\newblock Optimal decremental connectivity in planar graphs.
\newblock In {\em 32nd International Symposium on Theoretical Aspects of
  Computer Science, {STACS} 2015, March 4-7, 2015, Garching, Germany}, pages
  608--621, 2015.

\bibitem{Menger:1927}
Karl Menger.
\newblock Zur allgemeinen kurventheorie.
\newblock {\em Fund. Math.}, 10:96--115, 1927.

\bibitem{NS17}
Danupon Nanongkai and Thatchaphol Saranurak.
\newblock Dynamic spanning forest with worst-case update time: adaptive, {L}as
  {V}egas, and {O}(n\({}^{\mbox{1/2 - {\(\epsilon\)}}}\))-time.
\newblock In {\em Proceedings of the 49th Annual {ACM} {SIGACT} Symposium on
  Theory of Computing, {STOC} 2017, Montreal, QC, Canada, June 19-23, 2017},
  pages 1122--1129, 2017.

\bibitem{NSW17}
Danupon Nanongkai, Thatchaphol Saranurak, and Christian Wulff{-}Nilsen.
\newblock Dynamic minimum spanning forest with subpolynomial worst-case update
  time.
\newblock In {\em Proceedings of the 58th Annual Symposium on Foundations of
  Computer Science, {FOCS} 2017}, 2017.
\newblock To appear.

\bibitem{RZ08}
Liam Roditty and Uri Zwick.
\newblock Improved dynamic reachability algorithms for directed graphs.
\newblock {\em {SIAM} J. Comput.}, 37(5):1455--1471, 2008.

\bibitem{NYAS:NYAS340}
Pierre Rosenstiehl.
\newblock Embedding in the plane with orientation constraints: The angle graph.
\newblock {\em Annals of the New York Academy of Sciences}, 555(1):340--346,
  1989.

\bibitem{S04}
Piotr Sankowski.
\newblock Dynamic transitive closure via dynamic matrix inverse (extended
  abstract).
\newblock In {\em 45th Symposium on Foundations of Computer Science {FOCS}
  2004, 17-19 October 2004, Rome, Italy, Proceedings}, pages 509--517, 2004.

\bibitem{S16}
Shay Solomon.
\newblock Fully dynamic maximal matching in constant update time.
\newblock In {\em {IEEE} 57th Annual Symposium on Foundations of Computer
  Science, {FOCS} 2016, 9-11 October 2016, Hyatt Regency, New Brunswick, New
  Jersey, {USA}}, pages 325--334, 2016.

\bibitem{Sub-ESA-93}
Sairam Subramanian.
\newblock A fully dynamic data structure for reachability in planar digraphs.
\newblock In {\em Algorithms - {ESA} '93, First Annual European Symposium, Bad
  Honnef, Germany, September 30 - October 2, 1993, Proceedings}, pages
  372--383, 1993.

\bibitem{Thorup00}
Mikkel Thorup.
\newblock Near-optimal fully-dynamic graph connectivity.
\newblock In {\em Proceedings of the Thirty-Second Annual {ACM} Symposium on
  Theory of Computing, May 21-23, 2000, Portland, OR, {USA}}, pages 343--350,
  2000.

\bibitem{T05}
Mikkel Thorup.
\newblock Worst-case update times for fully-dynamic all-pairs shortest paths.
\newblock In {\em Proceedings of the 37th Annual {ACM} Symposium on Theory of
  Computing, Baltimore, MD, USA, May 22-24, 2005}, pages 112--119, 2005.

\bibitem{Nilsen13}
Christian Wulff{-}Nilsen.
\newblock Faster deterministic fully-dynamic graph connectivity.
\newblock In {\em Proceedings of the Twenty-Fourth Annual {ACM-SIAM} Symposium
  on Discrete Algorithms, {SODA} 2013, New Orleans, Louisiana, USA, January
  6-8, 2013}, pages 1757--1769, 2013.

\bibitem{W17}
Christian Wulff{-}Nilsen.
\newblock Fully-dynamic minimum spanning forest with improved worst-case update
  time.
\newblock In {\em Proceedings of the 49th Annual {ACM} {SIGACT} Symposium on
  Theory of Computing, {STOC} 2017, Montreal, QC, Canada, June 19-23, 2017},
  pages 1130--1143, 2017.

\end{thebibliography}
\end{document}